\documentclass[lettersize,journal]{IEEEtran}

\usepackage{amsmath,amssymb,amsfonts}
\usepackage{graphicx}
\usepackage{cite}
\usepackage{amsthm}
\usepackage[ruled,vlined]{algorithm2e}
\usepackage{algpseudocode}
\usepackage{xcolor}
\usepackage{soul}
\usepackage{tikz}
\usetikzlibrary{shapes.geometric, arrows.meta, positioning, calc}
\usepackage{algpseudocode}
\usepackage{booktabs}
\usepackage{tabularx}
\usepackage{makecell}
\usepackage{ragged2e}
\usepackage{array}
\usepackage{siunitx}
\usepackage{subfigure}
\usepackage{ragged2e}

\usepackage{caption}

\usepackage{fontawesome5} 
\usepackage{bm} 
\usepackage[hidelinks]{hyperref} 

\newcommand{\vect}[1]{\bm{#1}}

\newtheorem{theorem}{Theorem}

\newtheorem{assumption}{Assumption}

\theoremstyle{definition}

\newtheorem{rem}{Remark}

\algrenewcommand\algorithmicrequire{\textbf{Input:}}
\algrenewcommand\algorithmicensure{\textbf{Output:}}
\newcolumntype{L}[1]{>{\RaggedRight\arraybackslash}p{#1}}
\newcolumntype{Y}{>{\RaggedRight\arraybackslash}X}

\begin{document}

\title{EMS-FL: Federated Tuning of Mixture-of-Experts in Satellite-Terrestrial Networks via Expert-Driven Model Splitting}

\author{Angzi~Xu,
        Zezhong~Zhang,~\IEEEmembership{Member, IEEE},
        Zhi~Liu
        and Shuguang~Cui,~\IEEEmembership{Fellow,~IEEE}
}


\maketitle

\begin{abstract}
The rapid advancement of large AI models imposes stringent demands on data volume and computational resources.  Federated learning, though designed to exploit distributed data and computational resources, faces data shortage from limited network coverage and computational constraints from edge devices. To address these issues, both the mixture-of-experts (MoE) and satellite-terrestrial network (STN) provide promising solutions, offering lightweight computation overhead and broad coverage, respectively. However, the satellite-ground relative motion results in intermittent connectivity, hindering conventional federated learning that relies on model synchronization across devices. 
To leverage the coverage of STN while preserving training efficiency, we propose EMS-FL, an expert-driven model splitting and federated learning method. EMS-FL assigns each device cluster only the experts highly correlated to their local data. Through non-overlapping expert assignments, asynchronous local learning is further proposed, where each device cluster trains its assigned experts consecutively and only uploads local parameters to the satellite during connected phases for aggregation and model updates. Consequently, EMS-FL effectively reduces the training overhead and achieves both faster convergence and higher accuracy compared with conventional federated learning. Rigorous convergence analysis is provided to theoretically characterize the learning performance. Furthermore, comprehensive experiments are conducted using public datasets and large models, validating the superiority of EMS-FL.

\end{abstract}


\section{Introduction}\label{sec:introduction}
The rapid advancement of artificial intelligence (AI) has driven an explosive growth of intelligent applications, which are increasingly being deployed at the network edge with the development of edge mobile computing and edge learning \cite{Mao2017MEC,Letaief2022EdgeAI6G}. In recent years, the emergence of large-scale models has once again reshaped the industry landscape, owing to their remarkable capabilities and potential in handling tasks across a wide range of domains \cite{Brown2020GPT3,OpenAI2023GPT4,Chowdhery2023PaLM}. However, as AI models continue to grow in size, it raises much more stringent demands on data and computation resources for model training and tuning. For instance, the training process of Llama-3 involves approximately 15 trillion tokens and $10^{25}$ Flops of computation capacity, which requires about one year of training on ten thousand A100 GPUs \cite{Meta2024Llama3,Chu2025Llama3ISCA}. 
Currently, large models are typically trained and tuned in a centralized manner, with data aggregated on high-performance GPU servers. However, as distributed data at the network edge grow exponentially, centralized approaches are no longer able to fully exploit the data due to privacy constraints, leading to model obsolescence in turn. Moreover, centralized training also fails to harness the computational resources distributed across edge devices. To address these limitations, federated learning (FL) is considered a favorable solution, which has been widely adopted for collaborative AI model training in distributed wireless networks \cite{McMahan2017FedAvg,Chen2021JointLearningCommFL,Zhu2021OneBitOTA}. By exchanging model parameters or gradients instead of uploading the raw data samples, FL enables distributed training without compromising data privacy. Nevertheless, when applied to large models, FL faces two practical challenges as described in the following. 

The first challenge is the limited local computation and communication resources. Since most mobile devices cannot handle even a single round training of large models, such a practical issue has hindered distributed large model training for years. Under such a circumstance, the emergence of mixture-of-experts (MoE), a specialized paradigm of large models, opens the door to practical distributed training \cite{Fedus2022Switch,Du2022GLaM}
. In MoE, each dense feed-forward layer in conventional large models is replaced with multiple expert networks, among which only a sparse subset is activated for each data sample \cite{Shazeer2017MoE,Fedus2022Switch}. This design enables lightweight training and inference, thereby stimulating the demand for distributed MoE training\cite{Xue2024WDMoE,Chen2025SlimCaching}. Another challenge is the limited coverage of cellular networks. Training large models requires data from a vast number of devices, while the architecture of current cellular networks is unable to support such massive access. Although hierarchical FL has been proposed to involve more devices \cite{You2023HierPersonalizedTWC,Aygun2024HierClusteringTWC}, the architecture grows increasingly complex as the number of participants rises, leading to high latency and heavy communication overhead. To address this issue, the satellite-terrestrial network (STN) offers a promising solution. Satellites inherently provide much broader coverage than terrestrial networks. Furthermore, advancements in Low-Earth-Orbit (LEO) and Very-Low-Earth-Orbit (VLEO) satellites have enabled direct connectivity between mobile devices and satellites \cite{Ntontin2025Space6GProcIEEE,Le2025RandomAccessCOMST}. Given these advantages, STNs become a highly suitable scenario for distributed training of large models, with satellites serving as the access points for model aggregation and coordination.

According to analysis above, federated MoE training over STNs represents a promising paradigm for distributed training of large models. However, despite its notable advantages, this paradigm still faces a key challenge, namely the intermittent connectivity caused by relative motion between satellites and ground, which is an inherent characteristic of LEO constellations \cite{Matthiesen2024SatFLNetwork,Ntontin2025Space6GProcIEEE}. Conventional FL relies on consecutive and synchronous parameter exchange between the server and devices so that stable and persistent connectivity is necessary \cite{McMahan2017FedAvg}. However, in STNs, the high-speed relative motion between LEO satellites and ground devices prevents each satellite from maintaining a sustained connection with a fixed set of devices, which thereby disables conventional FL. To tackle this issue, recent research focuses on exploiting Inter Satellite Links (ISLs) and the Space-Air-Ground Integrated Network (SAGIN) \cite{Razmi2024OnboardISLTCOM,Fang2023OliveBranchTWC, Shi2024SatFEEL,Zhai2024FedLEO, Huang2024HFL_SAGIN}. 
Specifically, existing ISL-enabled solutions are typically targeted for on-board FL across LEO satellites, where the data, such as remote sensing and hyperspectral images, are collected and stored at satellites. In such architectures, a terrestrial ground station may serve as the central server for model aggregation, or even no central server is required \cite{Shi2024SatFEEL}. Although ISL-enabled solutions are effective for on-board FL across satellites, they do not account for the scenario where data are distributed at ground devices and thereby infeasible\cite{Razmi2024OnboardISLTCOM}. 
On the contrary, the SAGIN-based approaches aim to enhance FL across ground devices by exploiting the air nodes (e.g., unmanned aerial vehicles, airships, balloons, etc.) along with satellites for model aggregation\cite{Fang2023OliveBranchTWC, Huang2024HFL_SAGIN}. For instance, the authors in \cite{Fang2023OliveBranchTWC} propose to exploit the topology of SAGIN for hierarchical FL, where the Ring-All-Reduce algorithm is applied to enable fast inter-satellite parameter exchange.
However, these approaches require model synchronization and usually involve topology optimization to enable effective and efficient model aggregation in such a hierarchical architecture. Due to the intricate topology of SAGIN, parameter exchange of large models can result in severe congestion and high communication latency. Therefore, the existing methods based on ISLs and SAGIN are still inadequate for federated MoE training given the model scale.

Moreover, though ISLs and SAGIN can assist in model aggregation and update, such approaches usually involve multiple, or even all, satellites in training a single AI model. 
This could result in significant resource contention and network congestion once concurrent training of multiple large models is demanded, which is a foreseeable trend in the near future. Therefore, enabling a single satellite to manage the training of one or multiple AI models remains a critical requirement. 
To this end, we further dive deeply into the characteristics of the MoE model and FL framework. In fact, aside from a limited number of shared experts, most experts in an MoE model are specialized and will only be activated for their corresponding data modalities \cite{Shazeer2017MoE,Fedus2022Switch}. Moreover, each device, or a cluster of devices, located within a given area, typically owns only a narrow range of data types, such as financial language data in banking scenarios or obstacle images from vehicles on roads \cite{Guo2020PFLMoE}. Consequently, local training only activates a small subset of experts at each device or device cluster \cite{Fedus2022Switch,Lepikhin2021GShard,Du2022GLaM}. 
In view of the above facts, we first propose an expert-driven model splitting algorithm that assigns each device cluster the experts highly correlated to their local data. Constraints are also introduced to ensure non-overlapping expert assignments across all device clusters, based on which asynchronous local training is further proposed to enable continuous local expert training and updates, regardless of device connectivity to the satellite. This approach exploits the intrinsic parallelism of MoE models, where experts operate independently and can thus be trained in parallel by different device clusters, fully leveraging distributed computational resources. We call the proposed integration of \emph{expert-driven model splitting and asynchronous federated learning} EMS-FL. The main contributions are summarized as follows.

\begin{itemize}
    \item \textbf{STN-Assisted FL for Large Models:} Although FL has become prominent for distributed training of AI models, the rapid expansion in model scale is making data scarcity the critical bottleneck. To tackle this issue, we first introduce a STN-assisted FL framework where a LEO satellite serves as the access point, facilitating massive access and model aggregation. Such a framework leverages the broad coverage of LEO satellites while suffering from its inherent intermittent connectivity, representing the two facets of both the opportunity and the challenge. To the best of our knowledge, this work presents the first systematic investigation of distributed large model training under STN scenarios and proposes an efficient solution EMS-FL, which effectively mitigates the impact of intermittent links while capitalizing on the broad coverage. 
    
    \item \textbf{Lightweight Expert-Driven Model Splitting:} In conventional FL, every participating device is required to load the entire AI model for training, which turns out to be infeasible in large model training given limited local resources. Inspired by the observation that MoE only activates a small subset of experts for a given data modality, we propose a non-overlapping expert-driven model-splitting algorithm where each device cluster is assigned the highly correlated experts and only focuses on their training and updating. In this way, gradients of unassigned experts and shared model components can be discarded during training. Moreover, a masked gating mechanism is further proposed in enhanced EMS-FL, which directs inputs only to the assigned experts during local training, eliminating the need to load the unassigned experts into memory. As a result, the proposed model-splitting algorithm significantly reduces the resource consumption, enabling lightweight and practical local training on resource-constrained devices.    

    \item \textbf{High-Efficiency Training under Intermittent Links:} Due to the model synchronization requirement in conventional FL that each training round begins with a shared global model, the intermittent connectivity in STN hinders the local training at disconnected devices, severely degrading the training efficiency \cite{Matthiesen2024SatFLNetwork}. To address this issue, EMS-FL further proposes asynchronous and non-overlapping local expert training across device clusters, regardless of device connectivity to the satellite, fully exploiting the distributed computational resources and thereby significantly enhancing the training efficiency. Moreover, mathematical analysis is provided to quantify the convergence rate of the proposed EMS-FL, theoretically demonstrating its superiority over a conventional FL baseline. Comprehensive simulations are also conducted using public datasets and large models to verify the effectiveness and advantages of EMS-FL.
\end{itemize}

The reminder of this paper is organized as follows. The system model is described  in Section \ref{system-model}.  Section \ref{sec:EMS-FL} presents the proposed EMS-FL design. In Section \ref{sec:convergence}, the convergence performance of EMS-FL is derived. Experimental results are given in Section \ref{sec:sim_results}, followed by conclusions in Section \ref{sec:conclusion}.

\section{System Model}\label{system-model}
We consider an FL task for large model tuning in an STN comprising numerous ground devices and an LEO satellite. The objective is to fine-tune an MoE model by exploiting data distributed across all devices.  In the following, we present both the learning and satellite-ground communication models.



\subsection{Learning Model}\label{learning-model}
In this work, we consider FL under the coordination of a single server, i.e., an LEO satellite, as shown in Fig. \ref{systemmodel_fig}. The involved devices are naturally partitioned into $C$ clusters, each with $J$ devices, based on their geographic locations, denoted as $\{\mathcal{U}_1, \dots, \mathcal{U}_C\}$. The cluster indices form a set $\mathcal{C} = \{1,\dots,C\}$. Without loss of generality, we assume that the satellite stays in high-rate communication with only one device cluster within its instantaneous coverage at each time slot $t$, given as $\mathcal{U}(t)$ with $\mathcal{U}(t) = \mathcal{U}_c$ if $c = t - \lfloor \frac{t}{C} \rfloor C$, where function $\lfloor x \rfloor$ gives the greatest integer no larger than $x$. The local dataset at device $j$ from cluster $c$ is denoted as $\mathcal D_{c,j}$, containing $N$ data samples. The aggregation of all local datasets constitutes a global dataset $\mathcal{D} = \bigcup\limits_{c,j} \mathcal{D}_{c,j}$.
In the following, we present the FL framework and its adaptation for MoE fine-tuning.

\subsubsection{Federated Learning}\label{fl-framework}
Federated learning is a method for training an AI model with parameters $\boldsymbol{\theta}$ by exploiting data samples distributed across multiple devices.
The local loss function at device ${j}$ from cluster $c$ is defined as
\begin{align}\label{eq:local-loss}
F_{c,j}(\vect\theta) \triangleq \frac{1}{N} \sum_{(\boldsymbol{x}_n,y_n)\in \mathcal{D}_{c,j}} f(\boldsymbol{\theta};\boldsymbol{x}_n,y_n),
\end{align}
where $f(\boldsymbol{\theta};\boldsymbol{x}_n,y_n)$ denotes the sample-wise loss, and $\boldsymbol{x}_n$ and $y_n$ denote the feature vector and label of data sample $n$, respectively. For convenience, we rewrite $f(\boldsymbol{\theta};\boldsymbol{x}_n,y_n)$ as $f_n(\boldsymbol{\theta})$. Then, the global loss function is defined as 
\begin{align}\label{eq:global-loss}
F(\vect\theta) &\triangleq \frac{1}{CJN}
    \sum_{c=1}^C \sum_{j\in\mathcal U_c} \sum_{n\in \mathcal{D}_{c,j}} f_n(\boldsymbol{\theta}).
\end{align} 
Combining the definitions in \eqref{eq:local-loss} and \eqref{eq:global-loss} yields the relation
\begin{align}
    F(\vect\theta) =   \frac{1}{CJ}
    \sum_{c=1}^C \sum_{j\in\mathcal U_c} F_{c,j}(\vect\theta).
\end{align}
The FL process comprises multiple communication rounds. Without loss of generality, we consider the model aggregation case, where each round $t$ involves the four procedures below.
\begin{itemize}
    \item \textbf{Model broadcast}: The server first broadcasts the latest parameters $\boldsymbol{\theta}_t$ to all devices. 
    \item \textbf{Local training and model update}: After receiving the current parameters $\boldsymbol{\theta}_t$, each device $j$ from cluster $c$ computes a local gradient of the loss function in \eqref{eq:local-loss} using its local dataset $\mathcal{D}_{c,j}$. The local gradient is expressed as
    \begin{align}\label{eq:local-grad}
    \mathbf{g}_{t,c,j} = \frac{1}{N} \sum\limits_{n\in \mathcal{D}_{c,j}}  \nabla_{\boldsymbol{\theta}_t} f_n(\boldsymbol{\theta}_t),
    \end{align}
    After local gradient computation, the local model can be updated by applying gradient descent, given as 
    \begin{align}\label{local-update}
        \boldsymbol{\theta}_{t+1,c,j} = \boldsymbol{\theta}_{t} - \eta \mathbf{g}_{t,c,j},
    \end{align}
    where $\eta>0$ denotes the step-size.
    \item \textbf{Model uploading and aggregation}: Suppose reliable communication can be achieved in the model uploading phase. Then the global model update can be obtained through aggregation at the server, given as
    \begin{align}\label{model-aggregation}
        \boldsymbol{\theta}_{t+1} = \frac{1}{CJ}
    \sum_{c=1}^C \sum_{j\in\mathcal U_c} \boldsymbol{\theta}_{t+1,c,j}.
    \end{align}
This finishes a single round training process.
\end{itemize}
With reliable communication between the server and devices, the FL method is capable of achieving equivalent training performance compared to centralized training \cite{McMahan2017FedAvg}. However, with the emergence of large models, the substantial computational and communication demands usually exceed the capabilities of distributed devices. The MoE models emerged recently provide a flexible solution based on a sparse activation strategy, wherein only a small fraction of experts are engaged and trained per round. The MoE model and its tuning process under the FL framework are described in the following.

\begin{figure}[t]
	\centering
	\includegraphics[scale=0.54]{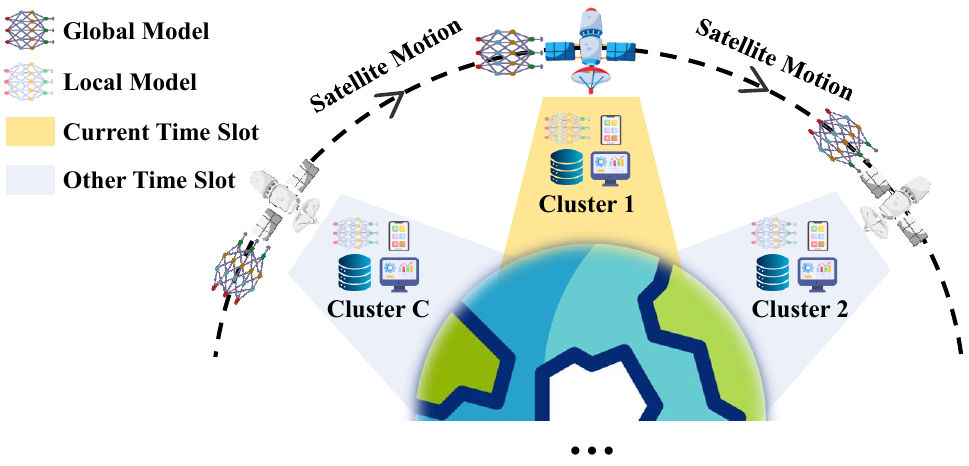}
	\caption{Illustration of large model tuning via FL under STN.}\label{systemmodel_fig}
\end{figure}

\begin{figure}[t]
  \centering
  \includegraphics[height=5cm]{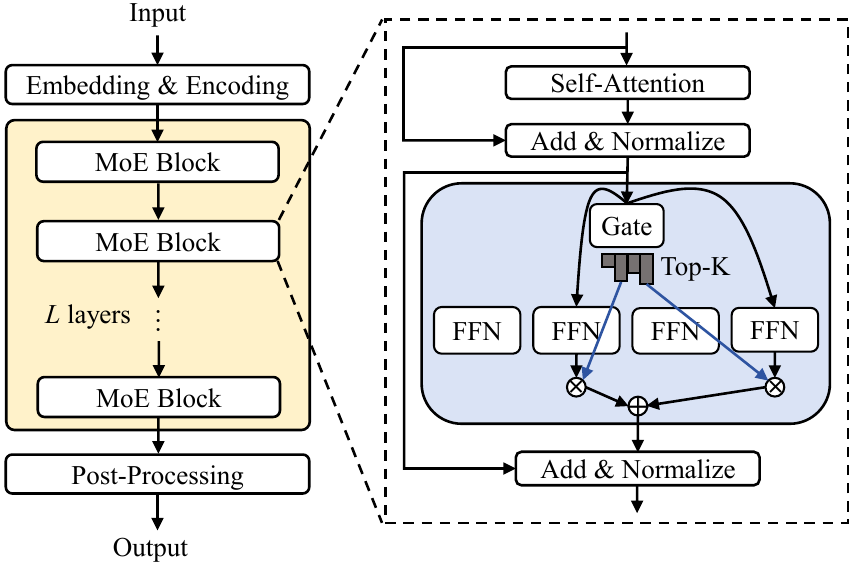}
  \caption{The architecture of transformer-based MoE\cite{Xue2024WDMoE}.}
  \label{MoE_fig}
\end{figure}

\subsubsection{MoE Fine-tuning via FL} 
An MoE model $Q$ is composed of $L$ MoE layers, where each layer $\ell$ consists of a gating network $\boldsymbol{u}_\ell$, a set of $M_\ell$ expert networks $\{\boldsymbol{e}_{\ell,m}\}_{m=1}^{M_\ell}$, and a shared backbone, as shown in Fig. \ref{MoE_fig}. In MoE fine-tuning, only the expert networks and the gating network are trainable.
Therefore, we omit the backbone parameters for simplicity. Given the input of MoE layer $\ell$ as $\boldsymbol{x}_\ell$, it first passes the gating network $\boldsymbol{u}_\ell$, outputting a vector $\mathcal{G}_\ell(\boldsymbol{u}_\ell;\boldsymbol{x}_\ell)$ with
\begin{align}
    \mathcal{G}_\ell(\boldsymbol{u}_\ell;\boldsymbol{x}_\ell)_m&=\operatorname{softmax}(g_\ell(\boldsymbol{u}_\ell;\boldsymbol{x}_\ell))_m,
\end{align}
where $g_\ell(\boldsymbol{u}_\ell;\boldsymbol{x}_\ell)$ represents the result before the softmax operation.
The input vector $\boldsymbol{x}_\ell$ is also fed into every expert $m$, producing $q_{\ell,m}(\boldsymbol{e}_{\ell,m}; \boldsymbol{x}_\ell)$. Then the output of layer $\ell$ is
\begin{align}\label{raw_moe}
    \!\!\!Q_{\!\ell}(\!\boldsymbol{e}_{\ell,1}, \dots, \boldsymbol{e}_{\ell,M_\ell}, \boldsymbol{u}_\ell;\boldsymbol{x}_\ell\!)\!=\!\!\sum_{m\!=\!1}^{M_\ell} \!q_{\ell,m}(\!\boldsymbol{e}_{\ell,m};\boldsymbol{x}_\ell\!) \mathcal{G}_\ell(\!\boldsymbol{u}_\ell;\boldsymbol{x}_\ell\!)_m .
\end{align}
To enable efficient training and inference while maintaining learning performance, model sparsification is commonly applied to each MoE layer, resulting in a \emph{sparse MoE} architecture \cite{Cai2025SurveyMoE}. In this setup, a Top-K routing strategy is employed at each gating network, given as
\begin{align}    \mathcal{G}_\ell(\boldsymbol{u}_\ell;\boldsymbol{x}_\ell)_m=\operatorname{softmax}\left(\operatorname{Top}\left(g_\ell(\boldsymbol{u}_\ell;\boldsymbol{x}_\ell)+\mathcal{E}, K\right)\right)_m,
\end{align}
where the $\operatorname{Top}(;,K)$ operator retains the top $K$ values of a vector and sets the others to $-\infty$. Then only $K$ out of the $M_\ell$ terms effectively contribute to the weighted sum in \eqref{raw_moe}. In subsequent presentation, the term MoE model refers to the \emph{sparse MoE} by default.

As the number of experts is typically consistent across layers, i.e., $M_\ell = M$, we group the expert networks sharing the same index $m$ into a single entity, referred to as expert $m$ and parameterized by $\boldsymbol{e}_m$ \cite{Qin2025OptimalExpertSelectionDMoE}. Moreover, we denote the aggregation of all experts as $\mathcal{M}$, parameterized by
\begin{align}
    \boldsymbol{e} \triangleq [\boldsymbol{e}_1^T,\dots,\boldsymbol{e}_M^T]^T.
\end{align}
The parameters of gating networks across all layers are also denoted collectively as $\boldsymbol{u}$ and termed the gate.
Then according to \eqref{raw_moe}, the model output can be described as a function of $\boldsymbol{e}$ and $\boldsymbol{u}$. For a given data sample $(\boldsymbol{x}_n,y_n)$, the model output is $Q(\boldsymbol{x}_n;\boldsymbol{e},\boldsymbol{u})$ and the sample-wise loss can be denoted by $f_n(\boldsymbol{e},\boldsymbol{u})$. By aggregating all trainable parameters as
\begin{align}
    \boldsymbol{\theta} &= [\boldsymbol{e}^T,\boldsymbol{u}^T]^T=[\boldsymbol{e}_1^T,\dots,\boldsymbol{e}_M^T,\boldsymbol{u}^T]^T,
\end{align}
the sample-wise loss can be written as $f_n(\boldsymbol{\theta})$, which aligns with the expression in Sec. \ref{fl-framework}. 
According to \eqref{eq:local-loss} and \eqref{eq:global-loss}, the local and global loss functions are given as
\begin{align}\label{eq:loss-baseline}
    F_{c,j}(\boldsymbol e,\boldsymbol u)
    &=  \frac{1}{N}
    \sum_{(\boldsymbol{x}_n,y_n)\in\mathcal D_{c,j}}
    f(\boldsymbol e,\boldsymbol u;\boldsymbol{x}_n,y_n),\\
    F(\boldsymbol e,\boldsymbol u)
    & =
    \frac{1}{CJ}
    \sum_{c=1}^C \sum_{j\in\mathcal U_c} F_{c,j}(\boldsymbol e,\boldsymbol u).
\end{align}
Then the MoE model can be tuned with distributed data under the FL framework following procedures in Sec. \ref{fl-framework}.

\begin{figure}[t]
  \centering
  \includegraphics[height=5cm]{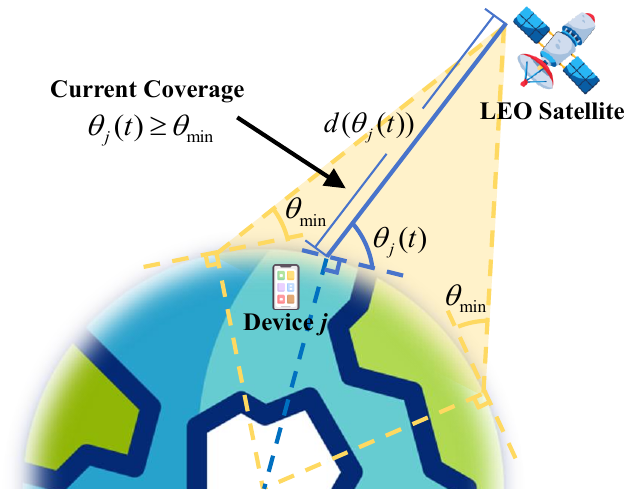}
  \caption{The instantaneous coverage of an LEO satellite.}
  \label{geometry_fig}
\end{figure}

\subsection{Satellite-Ground Communication Model}\label{subsec:comm}
By using the LEO satellite as the server, the gradient uploading and model downloading processes rely on the satellite-ground communication links. As shown in Fig. \ref{geometry_fig}, given the elevation angle between the satellite and device $j$ as $\theta_j(t)$, the channel for device $j$ can be modeled as 
\begin{align}\label{raw-channel}
    \!\!\!h_{\!j}(t)\!=\!\!\left\{\begin{array}{l}
    \!\!\!\!a(\theta_{\!j}(t)) e^{j \phi_{\text {rain }}}\!  e^{j 2 \pi \!f_{\!D}\!(\theta_{\!j}(t)) t}  \tilde{h}_{\!j}(t), \ \text{if}\ \theta_{\!j}(t)  \!\ge\! \theta_{\!\mathrm{min}},\\
    \!\!\!\!0,  \qquad\qquad\qquad\qquad\qquad\qquad\  \text{otherwise},
\end{array}\right.
\end{align}
where
\begin{align}\label{large-scale-fading}
    a(\theta_j(t))=\frac{\lambda}{4 \pi d(\theta_j(t))} \cdot 10^{-\frac{A(\theta_j(t))}{20}} \cdot \zeta,
\end{align}
gives the large-scale fading coefficient. 
In \eqref{large-scale-fading}, the term $d(\theta_j(t))$ represents the distance between the satellite and device $j$ dependent on the elevation angle $\theta_j(t)$. The atmospheric fading is denoted as $A(\theta_j(t))$, which is typically inversely proportional to the elevation angle, i.e., $A(\theta_j(t)) \propto 1 / \sin (\theta_j(t))$. The shadowing effect is denoted as $\zeta$. Furthermore, in the channel model \eqref{raw-channel}, $\theta_\mathrm{min}$ is a constant elevation threshold, where a device $j$ is considered connectable only when $\theta_j(t) \!\ge\! \theta_\mathrm{min}$. This captures the inherent intermittency characteristic of ground-satellite links. Consequently, when an FL task is considered in this scenario, the intermittent connectivity poses a major challenge to training efficiency. The terms $e^{j \phi_{\text {rain }}}$ and $e^{j 2 \pi f_D(\theta_j(t)) t}$ are phase shifts induced by rain and the Doppler effect, respectively. The Doppler frequency is given by
\begin{align}
    f_D(\theta_j(t))=\frac{v_r(\theta_j(t))}{\lambda},
\end{align}
where $\lambda$ is the carrier wavelength and $v_r(\theta_j(t))$ represents the radial velocity, which reaches its maximum when the satellite is near the horizon ($\theta_j(t) \approx 0^{\circ}$)  and drops to zero when the satellite is right above device $j$ ($\theta_j(t) = 90^{\circ}$). This characteristic coincides with the previous statement that high-rate transmission is only viable when the satellite is at high elevation angles, i.e. $\theta_j(t) \ge \theta_\mathrm{min}$. The small-scale fading $\tilde{h}_j(t)$ follows a Rician fading model.
Given the above channel model, we further consider the uplink transmission capability, which is the bottleneck of satellite-ground communication. The received signal at the satellite can be presented as 
\begin{align}
    y_j(t) = \sqrt{P_j G_\mathrm{sat}} h_j(t) s_j(t) + z_j(t),
\end{align}
where $s_j(t)\sim \mathcal{CN}(0,1)$ is the transmitted signal and $z_j(t) \sim \mathcal{CN}(0,\sigma_z^2)$ is the additional noise at the receiver. The factors $P_j$ and $G_\mathrm{sat}$ denote the transmit power of device $j$ and the antenna gain at the satellite, respectively. With a Rician channel $h_j(t)$, the ergodic channel capacity can be given as
\begin{align}\label{ergodic}
    C_{\text {ergodic}}&= B \cdot \mathbb{E}_t\left[\log _2(1+\frac{P_j G_\mathrm{sat} |h_j(t)|^2}{\sigma_z^2})\right],
\end{align}
with bandwidth $B$. With proper channel coding, the channel capacity can be evaluated using the Shannon capacity, which is the upper bound of the ergodic channel capacity, given as
\begin{align}
    C_{\text{upper}}= B \cdot \log _2\left(1+\frac{P_j G_\mathrm{sat} |a(\theta_j)|^2}{\sigma_z^2}\right).
\end{align}

According to the above elaboration, an uplink transmission rate of $5$ Mbps can be achieved with practical settings of a $23$ dBm transmit power, a $-160$ dB large-scale fading (600 km satellite-ground distance), a $40$ dBi antenna gain at the satellite, and a $5$ MHz sub-channel bandwidth. The connection phase lasts for $5\sim10$ minutes, making it possible to upload a file of $200$ MB. However, the current STN still faces the following challenges.
Firstly, the intermittency modeled in \eqref{raw-channel} prevents model synchronization across all devices. Hence only the connected devices can participate in the current training round, causing severe under-utilization of the distributed computation resources. Secondly, full-parameter tuning imposes intensive communication overhead considering the limited capacity of  STN links. Moreover, though the MoE architecture alleviates the computational overhead with sparse activation, each device still needs to load the entire model as data samples activates different experts, leading to high memory demands. The above issues urge designs to engage all devices in each training round and enable lightweight local training and parameter uploading. To this end, we propose EMS-FL in the next section.

\section{EMS-FL for STN-Assisted MoE Fine-Tuning}\label{sec:EMS-FL}

In this section, we first introduce a synchronous FL baseline and then present the detailed design of EMS-FL. An enhanced version of EMS-FL is also provided to facilitate partial model loading at devices by applying a masked gate in local training.

\subsection{Synchronous FL Baseline}\label{sec:baseline}
To facilitate FL with model synchronization, a straightforward idea is to only involve device cluster within the instantaneous coverage of the satellite in each training round. 
respectively. In line with the standard FL framework, each training round $t$ in the baseline scheme comprises three steps:
\begin{itemize}
    \item \textbf{Model broadcast}: The satellite first broadcasts
$\boldsymbol{\theta}_t = [\boldsymbol e_t^T,\boldsymbol u_t^T]^T$
to the connected device cluster with index $c^\prime = t - \lfloor \frac{t}{C} \rfloor C$, where the function $\lfloor x \rfloor$ gives the greatest integer no larger than $x$.
    \item \textbf{Local model update}: The local gradients of the experts and gate are denoted as 
    \begin{align}\label{baseline-gradient}
        \mathbf{g}_{t,c^\prime,j}^\mathrm{E} &=  \nabla_{\boldsymbol{e}_t} F_{c^\prime,j}(\boldsymbol{e}_t, \boldsymbol{u}_t), \\
        \mathbf{g}_{t,c^\prime,j}^\mathrm{U} &=  \nabla_{\boldsymbol{u}_t} F_{c^\prime,j}(\boldsymbol{e}_t, \boldsymbol{u}_t), 
    \end{align} 
    respectively.
    The deviation of the local gradient estimates from the global gradients are thereby given as 
    \begin{align}
        \boldsymbol{\delta}_{t,c^\prime,j}^{\mathrm{E}} &= \mathbf g^{\mathrm E}_{t,c^\prime,j} - \nabla_{\boldsymbol{e}_t} F(\boldsymbol{e}_{t},\boldsymbol u_{t}),\\
        \boldsymbol{\delta}_{t,c^\prime,j}^{\mathrm{U}} &= \mathbf g^{\mathrm U}_{t,c^\prime,j} - \nabla_{\boldsymbol u_t} F(\boldsymbol{e}_{t},\boldsymbol u_t).
    \end{align}     
    The local model update can be achieved by
    \begin{align}
    \boldsymbol{e}_{t+1,c^\prime,j} &= \boldsymbol{e}_t - \eta_t^\mathrm{E}\mathbf{g}_{t,c^\prime,j}^\mathrm{E},  \\
    \boldsymbol{u}_{t+1,c^\prime,j} &= \boldsymbol{u}_t - \eta_t^\mathrm{U}\mathbf{g}_{t,c^\prime,j}^\mathrm{U}.
    \end{align}
    \item \textbf{Model upload and aggregation}: With reliable uplink transmission, the global model update at the satellite can be achieved by model aggregation, given as
    \begin{align}
        \boldsymbol{e}_{t+1} \!=\! \frac{1}{J} \!\!\sum_{j\in\mathcal{U}_{c^\prime}} \!\boldsymbol{e}_{t+1,c^\prime,j}, \quad
        \boldsymbol{u}_{t+1} \!=\! \frac{1}{J} \!\!\sum_{j\in\mathcal{U}_{c^\prime}} \!\boldsymbol{u}_{t+1,c^\prime,j}.
    \end{align}  
\end{itemize}
Following the training approach above, the MoE model can be tuned until convergence after sufficient number of rounds. However, this baseline scheme suffers from two major drawbacks. Firstly, it engages only the devices within the instantaneous coverage of the satellite, failing to fully exploit the distributed computation and data resources. Secondly, it requires all devices to load the entire model, imposing high memory demands. To address these limitations, we propose EMS-FL, which splits the MoE model according to the specialization of experts, enabling asynchronous and lightweight local training. 

\begin{rem}[Model Update with LoRA]
    Low-rank-adaptation (LoRA) is a parameter-efficient fine-tuning technique that reduces communication overhead by decomposing model updates into low-rank matrices \cite{Hu2022LoRA}. Given limited capacity of STN links, we adopt LoRA to facilitate lightweight model aggregation. Taking the baseline scheme as an example, the model update $\Delta \boldsymbol{e}_{t,c^\prime,j} = \boldsymbol{e}_{t+1,c^\prime,j} - \boldsymbol{e}_t$ is factorized into two low-rank matrices and uploaded to the satellite. Then the local model $\boldsymbol{e}_{t+1,c^\prime,j}$ is reconstructed at the satellite for aggregation. In this work, LoRA is applied by default for model uploading.
\end{rem}

\subsection{EMS-FL for MoE}\label{sec:proposed}
The idea of EMS-FL is motivated by the specialization of experts in MoE models, where individual experts are tailored to specific data modalities. 
Specifically, since MoE employs the Top-K routing strategy, the experts tend to be functionally specialized once trained on multi-modal data. For a well-trained MoE, each expert is only activated when dealing with specific data modalities and remains inactive otherwise. In light of this fact, the EMS-FL design begins with an expert-driven model splitting algorithm, as described in the following.

\begin{algorithm}[t]
\caption{Expert-Driven Model Splitting}\label{model-split}
\KwIn{Model parameters $\boldsymbol{\theta}_t$, truncated expert relevant probability $\{\tilde{p}_{m,c}, \forall m, \forall c\}$}
\KwOut{Expert sets $\{\mathcal{M}_1,\mathcal{M}_2,\dots,\mathcal{M}_C\}$}
Initialize the feasible set as $\mathcal{C}_\mathrm{feasible} = \mathcal{C}$ and $C+1$ empty sets including $\mathcal{M}_1,\mathcal{M}_2,\dots,\mathcal{M}_C$ and $\mathcal{M}_\mathrm{full}$;

\For{ $m= 1,2,\dots$ }{

    Compute expert assignment probability $\!\{p_{m,c}^\mathrm{assign}\!\}_{\!c=1}^{\!C}$;
    
    Form a distribution $p(c) = p_{m,c}^\mathrm{assign}$, $c=1,\dots, C$;
    
    \If{$\boldsymbol{e}_m \notin \mathcal{M}_\mathrm{full}$}{
    
        Randomly pick $c^\star$ from $\mathcal{C}$ with $p(c) \neq 0$;
        
        \If{$c^\star \in \mathcal{C}_\mathrm{feasible}$ \textbf{or} $\sum_{c\in \mathcal{C}_\mathrm{feasible}} p(c) = 0$}{
        
        $\mathcal{M}_{c^\star} \leftarrow \mathcal{M}_{c^\star} \cup \boldsymbol{e}_m$, $\mathcal{M}_\mathrm{full} \leftarrow \mathcal{M}_\mathrm{full} \cup \boldsymbol{e}_m$;
        
        $\mathcal{C}_\mathrm{feasible} \leftarrow \mathcal{C}_\mathrm{feasible} \setminus \{c^\star\}$ if $|\mathcal{M}_{c^\star}| \ge K$;
        }
    }
}
\textbf{return} expert sets $\{\mathcal{M}_1,\mathcal{M}_2,\dots,\mathcal{M}_C\}$;
\end{algorithm}

\textbf{Expert-Driven Model Splitting}: The idea of expert-driven model-splitting is to assign each device cluster the experts that are strongly correlated with its local data, whereas the remaining parts of the MoE model are kept frozen. To this end, a representative \emph{trial dataset} $\mathcal{D}_c^\mathrm{trial}$ containing $N^\mathrm{trial}$ data samples is firstly collected from each device cluster $c$, which reflects the local data distribution. By feeding all the $N^\mathrm{trial}$ data samples into the current MoE model, we first compute the expert relevant probability $p_{m,c}$, defined as the probability that the gating network routes to expert $m$ at any MoE layer. For example, if expert $m$ is selected in $L_{m} (0 \le L_{m} \le L \times N^\mathrm{trial})$ MoE layers in total, the expert relevant probability is given by
\begin{align}\label{selection-probability}
    p_{m,c} = \frac{ L_{m} }{L \times N^\mathrm{trial}}. 
\end{align}
To identify the strongly correlated experts for device cluster $c$, a threshold-based truncation operation is applied as 
\begin{align}\label{truncation}
\tilde{p}_{m, c}=\left\{\begin{array}{ll}p_{m, c}, \quad& p_{m, c} \geq p_{\mathrm{th}} \\
0, \quad& p_{m, c}<p_{\mathrm{th}}
\end{array},\right.
\end{align}
given a threshold $p_\mathrm{th}$.
The \emph{expert assignment probability} is then defined as
\begin{align}\label{assignment-probability}
    p_{m,c}^\mathrm{assign} = \frac{\tilde{p}_{m, c}}{\sum_{c=1}^{C}\tilde{p}_{m, c}},
\end{align}
which naturally satisfies the normalization condition $\sum_{c=1}^{C}p_{m,c}^\mathrm{assign} = 1$. For each expert $m$, it is uniquely assigned to one device cluster through a probabilistic selection based on the distribution of $\{p_{m,c}^\mathrm{assign}\}_{c=1}^C$. We also set some restrictions in the expert selection process to ensure the experts are assigned to as many device clusters as possible. The details are provided in Algo. \ref{model-split}. The \emph{expert group} assigned to device cluster $c$ is denoted as $\mathcal{M}_c$ and parameterized as $\tilde{\boldsymbol{e}}_c$, containing $M_c$ experts. The rest $M-M_c$ experts are also loaded while frozen during training. This expert-driven model-splitting algorithm ensures that all experts are assigned non-overlappingly across the $C$ clusters, which thereby gives
\begin{align}
    \boldsymbol{e} = [\tilde{\boldsymbol{e}}_1^T,\dots,\tilde{\boldsymbol{e}}_C^T]^T.
\end{align}
The model-splitting operation is conducted periodically so that the experts can be trained by all the strongly-correlated device clusters in turn. Note that the model-splitting operation tends to assign the experts to device clusters with more correlated data. We define two ratio parameters as follows.
\begin{itemize}
    \item \emph{Local Expert Relevant Ratio} $\alpha_c$: Defined as the fraction of data samples in $\mathcal{D}_c$ that are correlated to $\tilde{\boldsymbol{e}}_c$.
    \item \emph{Global Expert Relevant Ratio} $\beta_c$: Defined as the fraction of data samples in  $\mathcal{D}$ that are correlated to $\tilde{\boldsymbol{e}}_c$. 
\end{itemize}
Nore that the information $\{(\alpha_c, \beta_c),\forall c\}$ is only defined for convergence performance analysis in Sec. \ref{sec:convergence} and unknown in practice. We also assume that the data samples correlated to each expert group $\tilde{\boldsymbol{e}}_c$ are drawn from the same source distribution while stored non-uniformly across device clusters, thereby creating heterogeneity in the data distribution.

\begin{rem}[Expert Relevant Probability]\label{remark-expert-relevance}
For each data modality, e.g., the dataset at a device cluster $c$, the expert relevant probability 
 $p_{m,c}$ typically exhibits a skewed distribution where only a small subset of experts are strongly correlated with the data \cite{Fedus2022Switch,Zoph2022STMoE}. Moreover, the strongly correlated experts can be categorized into two classes. One is \emph{shared experts} which are highly activated for almost all kinds of data. Another is \emph{domain-specific experts} which are activated only for a specific data modality.  
 The remaining experts are seldom activated, usually with a probability below $1\%$ each. 
\end{rem}

\begin{algorithm}[!t]
\caption{Federated Training (EMS-FL)}\label{algo-training}
\KwIn{MoE model $\boldsymbol{\theta}_0$, step-sizes $\eta^\mathrm{E}, \eta^\mathrm{U}$, threshold $p_\mathrm{th}$, device sets $\{\mathcal{U}_c\}_{c=1}^C$}
\KwOut{Global parameters $\boldsymbol{u}_{t_\mathrm{max}}$ and $\{\boldsymbol{e}_{{t_\mathrm{max}},m}\}_{m=1}^M$}
Initialize $\{\mathcal{M}_c\}_{c=1}^C$ with Algo. \ref{model-split}, and $\boldsymbol{u}_0$, $\{\tilde{\boldsymbol{e}}_{0,c}\}_{c=1}^C$ according to $\boldsymbol{\theta}_0$;

\For{$t=0,1,\dots,t_\mathrm{max}$}{
  
    \For{$c=1,2\dots, C$}{
  
        \eIf{$c = t - \lfloor \frac{t}{C} \rfloor C$}{
      
        Upload local experts $\{\tilde{\boldsymbol{e}}_{t,c,j}\}_{j\in\mathcal{U}_c}$;
      
        Expert update by aggregation in \eqref{global-update};
      
        Model download and local initialization with $\tilde{\boldsymbol{e}}_{t,c,j} = \tilde{\boldsymbol{e}}_{t,c}$, $\boldsymbol{u}_{t,j} = \boldsymbol{u}_{t}$ for $j \in \mathcal{U}_c$;

        Local update to get $\{\tilde{\boldsymbol{e}}_{t+1,c,j}\}_{j \in \mathcal{U}_c}$ and $\{\boldsymbol{u}_{t+1,j}\}_{j \in \mathcal{U}_c}$;
      
        Local gate upload for $\{\boldsymbol{u}_{t+1,j}\}_{j \in \mathcal{U}_c}$;
      
        Gate update with \eqref{gate-globale-update} to get $\boldsymbol{u}_{t+1}$;

        Download gate $\boldsymbol{u}_{t+1}$ to $j \in \mathcal{U}_c$;
      
        }
        {
        Local expert update to get $\{\tilde{\boldsymbol{e}}_{t+1,c,j}\}_{j \in \mathcal{U}_c}$;
        }
    }
}
\textbf{return} global parameters $\boldsymbol{u}_{t_\mathrm{max}}$ and $\{\boldsymbol{e}_{{t_\mathrm{max}},m}\}_{m=1}^M$;
\end{algorithm}

\textbf{Asynchronous Local Training}: As the satellite passes each device cluster once per orbital round, each device cluster experiences two phases in each orbital round, i.e., a \emph{connected phase} and a \emph{disconnected phase}. To adapt to the intermittent connectivity of satellite-ground communication, we propose asynchronous local expert training. 
For presentation convenience, we assume that a connected phase only comprises a single round of training, and the satellite passes all device clusters continuously, i.e., an orbital round consists of $C$ connected phases. At time slot $t=1$, the satellite begins the training process from device cluster $c=1$. We take device cluster $c$ as an example in the following elaboration.

\begin{itemize}
    \item \emph{Connected Phase}: In training round $t = t_\mathrm{orbit}C+c$ where $t_\mathrm{orbit}\in \mathbb{Z}$, devices in cluster $\mathcal{U}_c$ are connected to the satellite. Local experts $\{\tilde{\boldsymbol{e}}_{t,c,j}\}_{j \in \mathcal{U}_c}$ are uploaded to the satellite for expert-wise aggregation to update the corresponding experts in the global model, given as 
    \begin{align}\label{global-update}
        \tilde{\boldsymbol{e}}_{t,c} = \frac{1}{J}\sum_{j \in \mathcal{U}_c} \tilde{\boldsymbol{e}}_{t,c,j},
    \end{align}
    The entire updated global model is then downloaded to all devices in $\mathcal{U}_c$. For experts and the gate, an initialization is performed with $\tilde{\boldsymbol{e}}_{t,c}$ and $\boldsymbol{u}_{t}$. 
    Subsequently, a round of local training is performed during the connected phase, resulting in $\{\tilde{\boldsymbol{e}}_{t+1,c,j}\}_{j \in \mathcal{U}_c}$ and $\{\boldsymbol{u}_{t+1,j}\}_{j \in \mathcal{U}_c}$. Then the local gates $\{\boldsymbol{u}_{t+1,j}\}_{j \in \mathcal{U}_c}$ are uploaded to the satellite for aggregation, given as
    \begin{align}\label{gate-globale-update}
        \boldsymbol{u}_{t+1} = \frac{1}{J}\sum_{j \in \mathcal{U}_c} \boldsymbol{u}_{t+1,j},
    \end{align}
    and downloaded to devices in $\mathcal{U}_c$.
    \item \emph{Disconnected Phase}: Local training is conducted to update the experts during the disconnected phase, which comprises $C-1$ training rounds. The gate and the rest part of the MoE model are kept frozen during the period.  
\end{itemize}
The above training strategy is detailed in Algo. \ref{algo-training} and also illustrated in Fig. \ref{alo_fig}. The main idea is summarized as follows. During an orbital round, experts in $\mathcal{M}_c$ are trained consecutively for $C$ rounds at device cluster $c$, covering both connected and disconnected phases, and then uploaded to the satellite for aggregation and updating. After all experts in $\mathcal{M}_c$ are updated and synchronized across device cluster $c$, the gate is also tuned with one training round during the connected phase. 

\begin{figure}[t]
  \centering
  \includegraphics[scale=0.35]{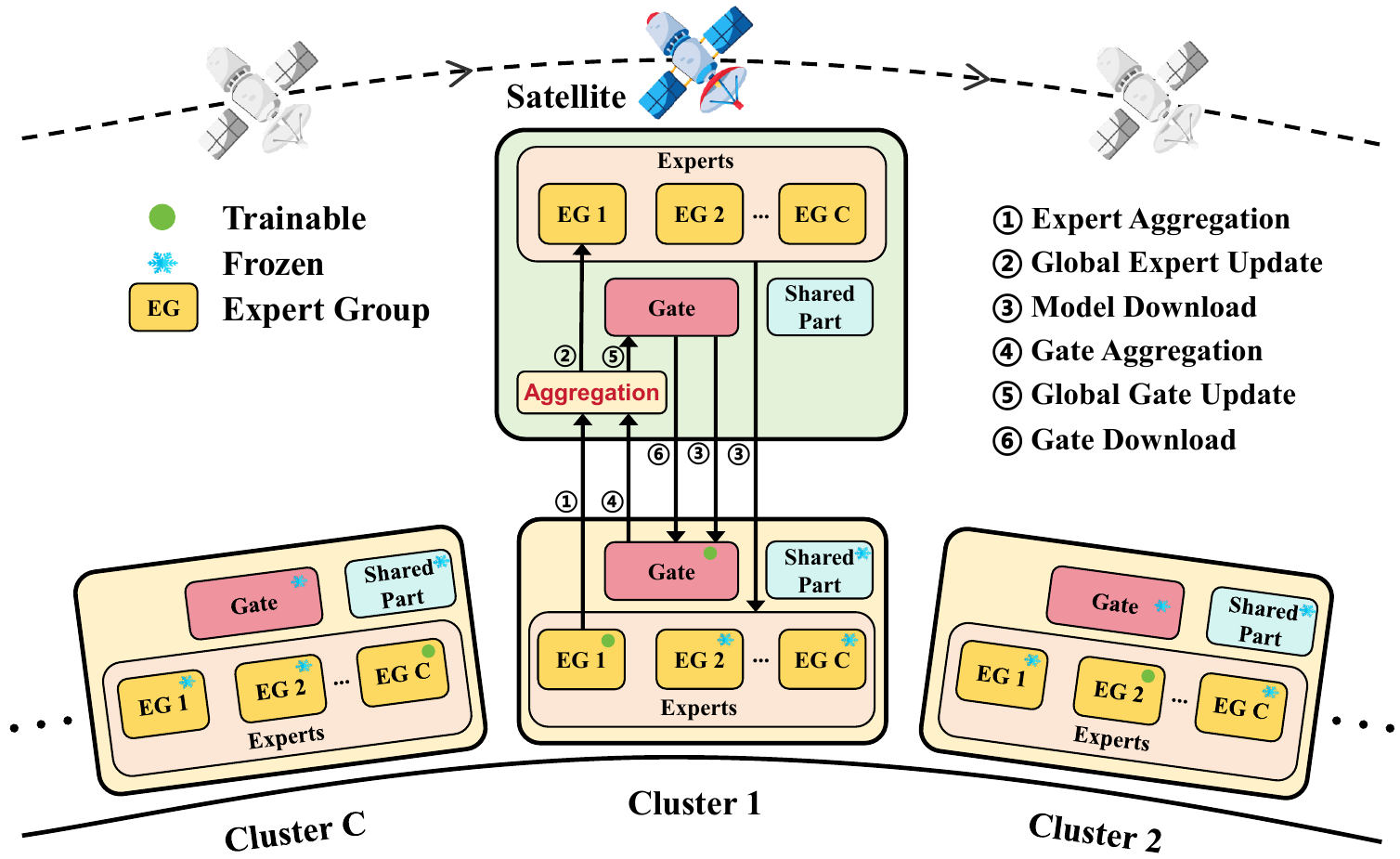}
  \caption{Asynchronous federated learning across device clusters.}
  \label{alo_fig}
\end{figure}

After introducing the framework of EMS-FL, the local training procedures on both experts and the gate are specified as follows.
The local loss at device $j$  from cluster $c$ is
\begin{align}\label{eq:local-loss-device}
    F_{c,j}(\boldsymbol e,\boldsymbol u)
    &=  \frac{1}{N}
    \sum_{(\boldsymbol x,y)\in\mathcal D_{c,j}}
    f(\boldsymbol e,\boldsymbol u;\boldsymbol x,y).
\end{align}
In training round $t$, the local gradient for expert group $\boldsymbol{\widetilde{e}}_{t,c,j}$ at device $j\in\mathcal U_c$ is given as
\begin{align}\label{eq:stochastic-gradients-E}
    \mathbf g^{\mathrm E}_{t,c,j}
    &= \nabla_{\boldsymbol{\widetilde{e}}_c}
    F_{c,j}(\boldsymbol{\widetilde{e}}_{t,c,j},\boldsymbol u_{t,c,j}),
\end{align}
and its deviation from the global gradient is written as
\begin{align}\label{eq:stochastic-gradients-E-error}
    \boldsymbol{\delta}_{t,c,j}^{\mathrm{E}} = \mathbf g^{\mathrm E}_{t,c,j} - \gamma_c \nabla_{\boldsymbol{\widetilde{e}}_c}
    F(\boldsymbol{\widetilde{e}}_{t,c,j},\boldsymbol u_{t,c,j}) .
\end{align}
with $\gamma_c = \frac{\alpha_c}{\beta_c}$. We use a scaled global gradient in \eqref{eq:stochastic-gradients-E-error} because the local density of data samples correlated to expert group ${\boldsymbol{\widetilde{e}}_c}$ is $\gamma_c$ times the corresponding density in the global dataset. For the gate, the gradient is only calculated in the round of the connected phase, given as
\begin{align}\label{eq:stochastic-gradients-U}
    \mathbf g^{\mathrm U}_{t,c,j}
    &=
    \nabla_{\boldsymbol u}
    F_{c,j}(\boldsymbol{e}_{t},\boldsymbol u_t), \quad \text{if} \ c = t - \lfloor \frac{t}{C} \rfloor C,
\end{align}
and its deviation from the global gradient is
\begin{align}\label{eq:stochastic-gradients-U-error}
    \boldsymbol{\delta}_{t,c,j}^{\mathrm{U}} = \mathbf g^{\mathrm U}_{t,c,j} - \nabla_{\boldsymbol u}
    F(\boldsymbol{e}_{t},\boldsymbol u_t), \quad  \text{if} \ c = t - \lfloor \frac{t}{C} \rfloor C.
\end{align}



\begin{rem}[Global Model Update]
The global model is referred to as the model at the satellite. In each training round $t$, the experts in the global model are partially updated by 
\begin{align}\label{global-expert}
\tilde{\boldsymbol{e}}_{t, c}= \begin{cases} \text{Eq. \eqref{global-update}}, & \text { if } c = t - \lfloor \frac{t}{C} \rfloor C, \\ \tilde{\boldsymbol{e}}_{t-1, c}, & \text { otherwise }.\end{cases}
\end{align}
The gate in the global model is updated according to Eq. \eqref{gate-globale-update}. Therefore, from the viewpoint of the satellite, the global model undergoes partial updates in every training round. After one complete orbital cycle, each expert in the global model is updated exactly once, whereas the gate is updated $C$ times.
\end{rem}

\begin{algorithm}[t]
\caption{Federated Training (Enhanced EMS-FL)}\label{algo-enhanced}
\KwIn{MoE model $\boldsymbol{\theta}_0$, step-sizes $\eta^\mathrm{E}, \eta^\mathrm{U}$, threshold $p_\mathrm{th}$, device sets $\{\mathcal{U}_c\}_{c=1}^C$}
\KwOut{Global parameters $\boldsymbol{u}$ and $\{\boldsymbol{e}_{{t_\mathrm{max}},m}\}_{m=1}^M$}
Initialize $\{\mathcal{M}_c\}_{c=1}^C$ with Algo. \ref{model-split}, $\boldsymbol{u}$, $\{\tilde{\boldsymbol{e}}_{0,c}\}_{c=1}^C$ according to $\boldsymbol{\theta}_0$, and $\boldsymbol{u}_c^\mathrm{mask} = \operatorname{masking}(\boldsymbol{u},\mathcal{M}_c)$;

\For{$t=0,1,\dots,{t_\mathrm{max}}$}{
  
    \For{$c=1,2\dots, C$}{
  
        \eIf{$c = t - \lfloor \frac{t}{C} \rfloor C$}{
      
        Upload local experts $\{\tilde{\boldsymbol{e}}_{t,c,j}\}_{j\in\mathcal{U}_c}$;
      
        Expert update by aggregation with \eqref{global-update};
      
        Model download and local initialization with $\tilde{\boldsymbol{e}}_{t,c}$ for $j \in \mathcal{U}_c$;

        Local update to get $\{\tilde{\boldsymbol{e}}_{t+1,c,j}\}_{j \in \mathcal{U}_c}$;
        }
        {
        Local expert update to get $\{\tilde{\boldsymbol{e}}_{t+1,c,j}\}_{j \in \mathcal{U}_c}$;
        }
    }
}
Train gate $\boldsymbol{u}$ with standard EMS-FL for a few rounds;

\textbf{return} global parameters $\boldsymbol{u}$ and $\{\boldsymbol{e}_{{t_\mathrm{max}},m}\}_{m=1}^M$;
\end{algorithm}

\subsection{Enhanced EMS-FL}\label{sec:enhanced}
Inspired by the idea of asynchronous local training in EMS-FL, this subsection presents a more radical design called the enhanced EMS-FL. We introduce a masking mechanism applied to the local gate, forcing it to route inputs only to the assigned expert group. This idea stems from the observation that the gate is typically well-trained or initialized with empirical knowledge, and can be easily tuned after expert training. Thus, the expert training will not be hampered by a frozen and masked gate. Instead, it offers two significant advantages. First, experts assigned to device cluster $c$ are activated and trained with higher probabilities compared with the standard EMS-FL as experts outside $\mathcal{M}_c$ will not be activated. Second, since the gate does not route to frozen experts, the local training can be further lightened by not loading them into memory. The training procedures of the enhanced EMS-FL are presented below and also summarized in Algo. \ref{algo-enhanced}.


Firstly, the model-splitting and expert assignment processes still follow Algo. \ref{model-split}. Then a masked gate is generated as 
\begin{align}\label{masking}
    \boldsymbol{u}_c^\mathrm{mask} = \operatorname{masking}(\boldsymbol{u},\mathcal{M}_c, \mathcal{M}),
\end{align}
for each device in cluster $c$, where the masking operation is to set the output logits of the gate corresponding to experts $m \notin \mathcal{M}_c$ to zero. The logits after masking are then passed through a $\operatorname{softmax}(\cdot)$ operator to obtain the weighting distribution, based on which the experts are selected via the Top-K routing strategy. 
By using the masked gate, each device cluster $c$ trains only the experts in $\tilde{\boldsymbol{e}}_c$ during both the connected and disconnected phases. The model aggregation is performed once during the connected phase in each orbital round to update the global expert $\tilde{\boldsymbol{e}}_{t,c}$. After the experts are well-trained through sufficient orbital rounds, the gate can be tuned by executing the standard EMS-FL for a few rounds. Compared to standard EMS-FL, the enhanced version improves local training efficiency and highly reduces the local memory demands. Although the introduction of the masked gate may lead to a more differentiated training of experts, this is not detrimental to the MoE model. On the contrary, the potential of each expert is fully unleashed with sufficient activation. In fact, conventional MoE models often suffer from the issue of load imbalance, resulting in a large number of experts being infrequently activated and under-trained. By introducing the masking operation, the enhanced EMS-FL constrains local training to be routed to the assigned experts $\tilde{\boldsymbol{e}}_c$ at each device cluster $c$, which inherently strengthens the load balance.

\section{Convergence Analysis}\label{sec:convergence}

In this section, we give a theoretical analysis on the convergence performance for the proposed EMS-FL and baseline schemes. Since the masking operation is non-differentiable, we confine our analysis to the standard EMS-FL and consider the enhanced version gives comparable performance. 

For tractable analysis on the convergence of EMS-FL, several widely-adopted assumptions are first given as follows.
\begin{assumption}[$L$-smoothness]\label{ass:smooth}   \emph{
The global loss function $F(\boldsymbol\theta)$ is differentiable and
$L$-smooth, i.e., there exists a constant $L>0$ such that 
\begin{align}\label{eq:global-L-smooth}
  \big\|\nabla F(\boldsymbol\theta_1)-\nabla F(\boldsymbol\theta_2)\big\|\le  L\big\|\boldsymbol\theta_1-\boldsymbol\theta_2\big\|, \quad \forall \boldsymbol\theta_1,\boldsymbol\theta_2. 
\end{align}
}
\end{assumption}

\begin{assumption}[Bounded Gradients]\label{ass:bounded-grad} \emph{
There exist constants $G_E>0$ and $G_U>0$ satisfying
\begin{align}\label{eq:bound-gradients-E}
  \|\nabla_{\boldsymbol{e}} F(\boldsymbol{e},\boldsymbol u)\| \le G_E, \quad \|\nabla_{\boldsymbol{u}} F(\boldsymbol{e},\boldsymbol u)\| \le G_U, \quad \forall \boldsymbol{e}, \boldsymbol{u}.
\end{align}
}
\end{assumption}

\begin{assumption}[Bounded Gradient Error in EMS-FL]\label{ass:var} \emph{
It is assumed that the stochastic gradients $\mathbf g^{\mathrm E}_{t,c,j}$ in \eqref{eq:stochastic-gradients-E} are independent and unbiased estimates of scaled global gradient $\gamma\nabla_{\boldsymbol{\widetilde{e}}_c} F\bigl(\boldsymbol{\widetilde{e}}_{t,c,j},\boldsymbol u_{t,c,j}\bigr)$, i.e.,
\begin{align}\label{expert-variance}
    &\mathbb{E}[\mathbf g^{\mathrm E}_{t,c,j}] = \gamma_c\nabla_{\boldsymbol{\widetilde{e}}_c} F(\boldsymbol{\widetilde{e}}_{t,c,j},\boldsymbol u_{t,c,j}), \quad\ \forall t,c,j, \\
    &\mathbb{E}[\|\boldsymbol{\delta}_{t,c,j}^{\mathrm{E}}\|^2]\ \le \gamma^2\sigma_E^2,  \qquad \qquad\qquad\  \forall t,c,j,
\end{align}
with constant $\sigma_E^2$ and $\gamma = \max\{\gamma_c, \forall c\}$. 
The stochastic gradients $\mathbf g^{\mathrm U}_{t,c,j}$ in \eqref{eq:stochastic-gradients-U} are unbiased estimates of global gradient $\nabla_{\boldsymbol u} F\bigl(\boldsymbol{e}_{t},\boldsymbol u_t\bigr)$ while correlated across devices, i.e.,
\begin{align}\label{noniid-gate-variance}
    \!\!\mathbb{E}[\mathbf g^{\mathrm U}_{t,c,j}] \!=\! \nabla_{\!\boldsymbol u} F(\boldsymbol{e}_{t},\boldsymbol u_t), \quad
    \mathbb{E}[\|\boldsymbol{\delta}_{t,c,j}^{\mathrm{U}}\|^2]\ \!\!\le\! \sigma_U^2, \quad \forall t,c,j, 
\end{align}
with constant $\sigma_U^2$.
}
\end{assumption}

Based on the above assumptions, the analytical convergence performance of EMS-FL is provided in the following. Since updating all experts and the gate of the global model takes one orbital cycle, we only evaluate the convergence performance at training rounds $t \in \{t_\mathrm{orbit} C, t_\mathrm{orbit} \in \mathbb{Z}\}$.

\begin{theorem}[Convergence of EMS-FL]
\label{thm:emsfl}
After $T$ orbital-cycle training with step-sizes $\eta^{\mathrm E} \le \frac{\eta^{\mathrm U}}{\gamma} \le  \frac{1}{\sqrt{T}}$ in EMS-FL, the gradient variance is upper bounded by
\begin{align}\label{eq:emsfl-theorem}
    &\frac{1}{T}  \sum_{t=0,C,\ldots,(T-1)C}  \mathbb E\|\nabla F(\boldsymbol\theta_t)\|^2 \nonumber \\
    \le& \frac{\mathbb E[F(\boldsymbol\theta_0) \!-\! F^\star]}{\gamma C\sqrt{T}}  \!+\! \frac{5L C^2\gamma[ C (G_E^2 \!+\! \sigma_E^2) \!+\! G_U^2\!+\!\sigma_U^2]}{\sqrt{T}}.
\end{align}
\end{theorem}
\begin{proof}
    See Appendix A.
\end{proof}

The results in Theorem \ref{thm:emsfl} shows that the expected gradient variance is bounded and asymptotically approaches $0$ after a sufficient number of orbital cycles $T$. The coefficient $\frac{1}{\gamma C\sqrt{T}}$ associated with the term $\mathbb E[F(\boldsymbol\theta_0) \!-\! F^\star]$ defines the convergence rate, which drives the model parameters towards the optimum. The factor $\gamma$ represents the ratio of the local density of expert-correlated data samples to the corresponding global density. As the experts are assigned to devices clusters that contains more correlated data, the factor $\gamma$ lies in a range $[1,C]$ dependent on the data heterogeneity, which correspondingly affects the convergence rate.  The second term quantifies the deviation caused by gradient errors, which also converges to zero asymptotically. For comparison, the analytical convergence performance of the baseline scheme is also derived below.

To facilitate convergence analysis for the baseline scheme, we first introduce Assumption \ref{ass:var-baseline} to replace the original assumptions on $\mathbf g^{\mathrm E}_{t,c,j}$ and $\boldsymbol{\delta}_{t,c,j}^{\mathrm{E}}$.
 in the baseline scheme are incompatible with that used in EMS-FL.
\begin{assumption}[Bounded Gradient Error in Baseline]\label{ass:var-baseline} \emph{
We assume the stochastic gradients $\mathbf g^{\mathrm E}_{t,c,j}$ in \eqref{baseline-gradient} are unbiased estimates of global gradient $\nabla_{\boldsymbol{e}_t} F(\boldsymbol{e}_{t},\boldsymbol u_{t})$ while correlated across devices, i.e.,
\begin{align}\label{noniid-data-variance}
    &\mathbb{E}[\mathbf{g}^{\mathrm E}_{t,c,j}] = \frac{1}{CJ}\sum\limits_{c=1}^{C}\sum\limits_{j\in\mathcal{U}_c} \mathbf{g}^{\mathrm E}_{t,c,j} =    \nabla_{\boldsymbol{e}_t} F(\boldsymbol{e}_{t},\boldsymbol u_{t}), \quad \forall t,\\
        &\mathbb{E}[\|\boldsymbol{\delta}_{t,c,j}^{\mathrm{E}}\|^2] \le \zeta_E^2+ \sigma_E^2, \quad \forall t, 
\end{align}
where $\sigma_E^2$ is defined in Algo. \ref{ass:var}, and $\zeta_E^2$ accounts for the variance components from data heterogeneity. }
\end{assumption}

Based on the assumptions above, the analytical convergence rate of the baseline scheme is derived below. 
\begin{theorem}[Convergence of Baseline]\label{thm:baseline}
After $T$ orbital-cycle training with step-sizes $\eta^{\mathrm E} \le  \frac{\eta^{\mathrm U}}{\gamma} \le  \frac{1}{\sqrt{T}}$ in the baseline scheme, the gradient variance is upper bounded by
\begin{align}\label{eq:emsfl-theorem}
\!\!&\frac{1}{T}\sum_{t=0,C,\ldots,(T-1)C}\mathbb E\|\nabla F(\boldsymbol\theta_t)\|^2 \nonumber\\
\!\!\le&
\frac{\mathbb E\!\left[F(\boldsymbol\theta_0)-F^\star\right]}{C\sqrt{T}} \nonumber\\
&+\frac{5LC^2(G_E^2+\zeta_E^2+\sigma_E^2)+5LC^2\gamma^2(G_U^2+\sigma_U^2)}{\sqrt{T}}.
\end{align}
\end{theorem}
\begin{proof}
    See Appendix B.
\end{proof}

Therefore, a comparison on theoretical convergence performance of EMS-FL and the baseline scheme is provided by Theorem \ref{thm:emsfl} and Theorem \ref{thm:baseline}. Both schemes achieve convergence after sufficient training, while at different rates, i.e., $\frac{1}{\gamma C\sqrt{T}}$ and $\frac{1}{C\sqrt{T}}$. Specifically, we further consider two extreme cases with $\gamma = 1$ and $\gamma = C$, corresponding to uniform data distribution and extreme heterogeneity where all data correlated with an given expert group reside in a single device cluster. When $\gamma = 1$, the convergence rates turn to be the same for both schemes, while EMS-FL suffers from a larger deviation in the second term due to  gradient errors. This is because each expert group is correlated to the data at all device clusters while only trained by a specific device cluster in EMS-FL, leading to under-utilization of the data resource and thereby causing slower convergence. When $\gamma = C$, EMS-FL achieves a noticeably faster convergence at rate $\frac{1}{C^2\sqrt{T}}$. Though in EMS-FL the deviation in the second term is amplified due to the scaling effect in \eqref{eq:stochastic-gradients-E-error}, the baseline scheme is penalized by an additional factor $\zeta_E^2$, which is large for severe data heterogeneity. Hence, EMS‑FL delivers a superior convergence rate in scenarios with highly heterogeneous data distribution. In fact, as clarified in Remark \ref{remark-expert-relevance}, it is usually the case in practice that each expert is correlated with specific data modalities, making $\gamma \gg 1$. And the data heterogeneity can be coarsely evaluated in advance using the expert assignment probability in \eqref{assignment-probability}, helping us choose the correct solution.

\begin{rem}[Idle Connection Slots]
    Note that we assume consecutive training across device clusters, where an orbital cycle exactly spans $C$ training rounds. In practice, however, there can be numerous idle connection slots where no device cluster lies in the coverage of the satellite. This further slows down the training process in the baseline scheme, while not affecting EMS-FL that applies asynchronous local training. Therefore, for the case with abundant idle connection slots, EMS-FL remains applicable and maintains high training efficiency. 
\end{rem}

\section{Simulation Results}\label{sec:sim_results}
In this section, we provide evaluations on the performance of the proposed EMS-FL through simulations. Two experiments are carried out with different settings, corresponding to MoE fine-tuning and training tasks, respectively. For EMS-FL, we use the enhanced version due to its lightweight properties.

\begin{table}[t]
\centering
\caption{Communication overhead and peak GPU memory with LoRA ($r=4$).}
\label{tab:cost}
\footnotesize
\setlength{\tabcolsep}{6pt}
\renewcommand{\arraystretch}{1.10}

\sisetup{
  detect-all,
  mode=text,
  table-number-alignment=center,
  table-text-alignment=center
}

\begin{tabular}{l
                S[table-format=2.2]
                S[table-format=2.2]
                S[table-format=2.2]}
\toprule
Method &

\multicolumn{1}{c}{Uplink Overhead (MB)} &
\multicolumn{1}{c}{Peak GPU Memory (GB)} \\
\midrule
EMS-FL  & 14.16 & 23.23 \\
Baseline                  & 43.25 & 46.00 \\
\midrule
\multicolumn{1}{l}{Ratio / Diff.} &
\multicolumn{1}{c}{$\approx 3.06\times$} &
\multicolumn{1}{c}{+22.77} \\
\bottomrule
\end{tabular}
\end{table}

\subsection{Experimental Settings}
\label{sec:sim_settings}
The default settings are as follows unless specified otherwise. We consider a scenario as described in Sec.~\ref{subsec:comm}. An LEO satellite serves as the server for model aggregation and global model updates. At each time slot, only one device cluster is connected to the satellite with an uplink bandwidth of 5 GHz in total, which can support 1000 devices each with a 5 Mbps transmission rate under the setting of a 23 dBm transmit power, a -160 dB large-scale fading (600 km communication distance), a 40 dBi antenna gain at the satellite, and a 5 MHz sub-channel bandwidth. This enables the uploading of a 300 MB file during a 10-minute connection phase. In our experiments, each expert is an 8B large model, which contributes a 10-20 MB gradient parameter file after using LoRA with rank 4.
The downlink transmission process is neglected due to the high-capacity broadcast capability of the satellite. The local expert training at each device cluster is simulated using a single 80G A100 GPU. 
In the following, we further present the training details and the corresponding simulation results of the two experiments. 

\begin{figure}[t]
        \centering

        \subfigure[Fine-tuning process with step-size $\eta^\mathrm{E} = 1\times10^{-5}$.]{\includegraphics[height=5cm, width=8cm]{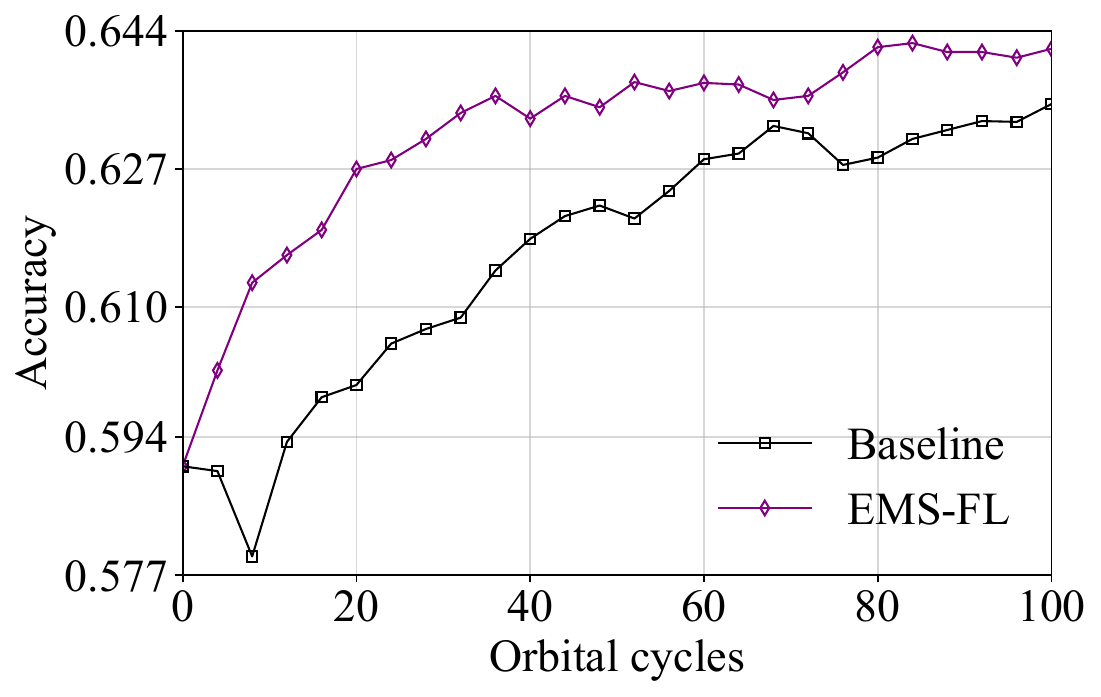}}
        \subfigure[Fine-tuning process with step-size $\eta^\mathrm{E} = 6\times10^{-6}$.]{\includegraphics[height=5cm, width=8cm]{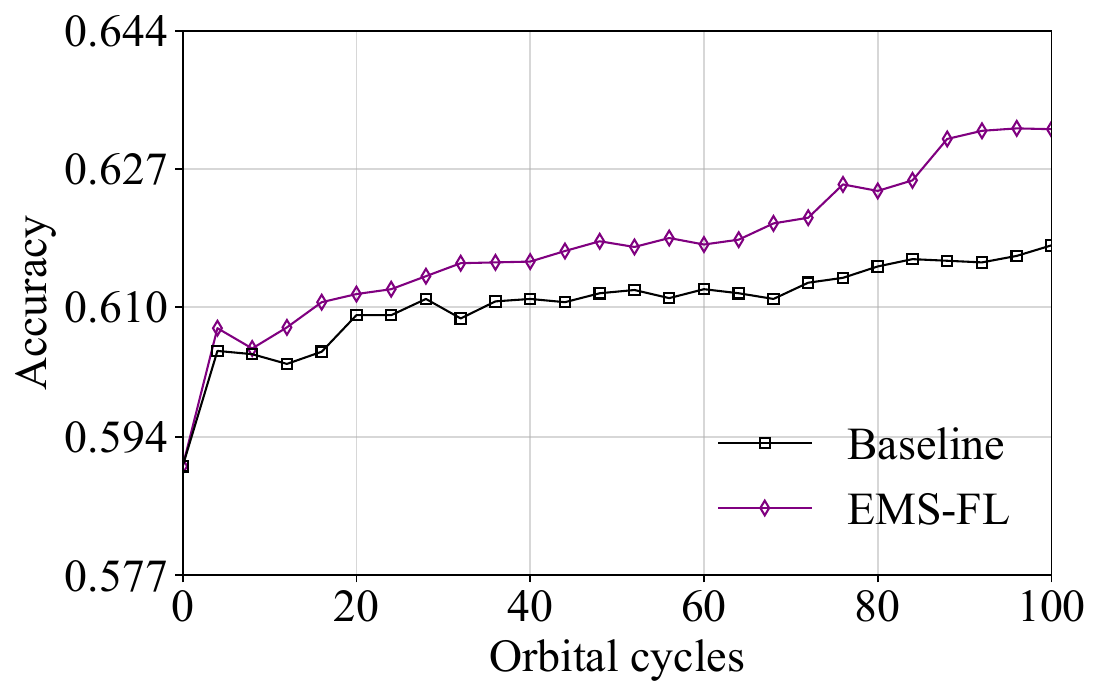}}
        \subfigure[Fine-tuning process with step-size $\eta^\mathrm{E} = 1\times10^{-6}$.]{\includegraphics[height=5cm, width=8cm]{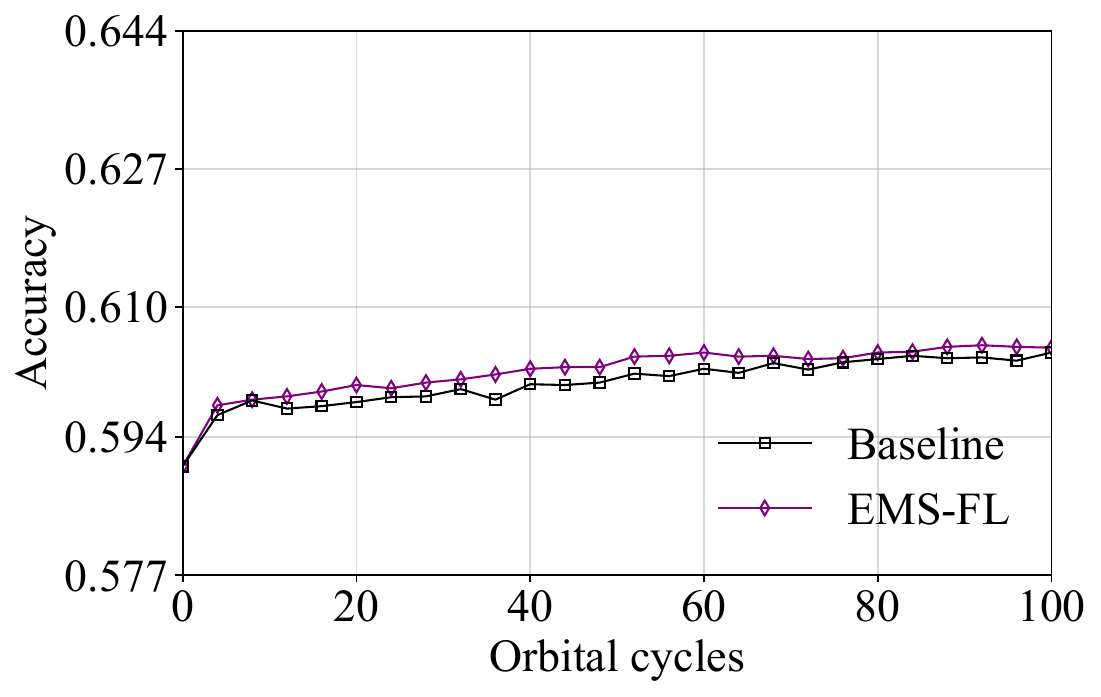}}
    \caption{Test accuracy versus number of orbital cycles for EMS-FL and baseline scheme under different step-sizes.}
    \label{ex1_dflr}
\end{figure}

\subsection{Experiment 1: MoE Fine-tuning with Top-1 Routing}\label{experiment1}
In this experiment, we consider a fine-tuning task where an MoE model is constructed by integrating three  pre-trained Llama-3-8B variants: Llama-3-8B-Instruct~\cite{meta_llama3_8b_instruct_2024},
Llama-3-8B-Chinese-Chat~\cite{wang_llama3_8b_chinese_chat_2024}, and
Llama-3-OpenBioLLM-8B~\cite{pal_llama3_openbiollm_8b_2024}. Specifically, the expert networks in MoE layer $\ell$ are initialized with feed-forward networks (FFNs) from the $\ell$-th transformer layers of the three Llama-3-8B models. A gate is introduced with Top-1 routing, which is randomly initialized and pre-trained for a few iterations (60 rounds) on a small set of shared data samples to ensure proper model functionality before tuning.
As we consider the case with $C \le M$, three device clusters are involved in this fine-tuning task. Furthermore, three open datasets, i.e., MMLU~\cite{hendrycks2021measuring} for general knowledge, CMExam~\cite{liu2023cmexam} for Chinese-specific medical knowledge, and MedMCQA~\cite{pal2022medmcqa} with English-specific medical knowledge, are used as local datasets at the three device clusters, respectively. 
Note that there is naturally a one-to-one mapping between the the experts and local datasets due to their inherent correlations. For example, the model Llama-3-8B-Chinese-Chat is highly correlated to the dataset CMExam and is thus assigned to the second device cluster according to the expert-driven model splitting algorithm.
All three local datasets are composed of multiple-choice question answering (MCQA) samples. In the tuning process, for each sample, we compute the length-normalized conditional log-likelihood of every option given the question prompt, selecting the option with the highest log-likelihood as the prediction \cite{Brown2020GPT3}. 
During fine-tuning, we employ cross-entropy loss between the predictions and the labels.

The convergence performance of both the EMS-FL and baseline schemes is shown in Fig. \ref{ex1_dflr}, trained with step-sizes $\eta^\mathrm{E}$ ranging from $10^{-6}$ to $10^{-5}$, illustrating the accuracy improvement process during $100$ orbital cycles ($300$ training rounds). Specifically, the accuracy is calculated with respect to the global dataset. In Fig. \ref{ex1_dflr}, the proposed EMS-FL always features faster convergence under different step-sizes and even attains a higher accuracy, which aligns with the convergence analysis in Sec. \ref{sec:convergence}, validating the superiority of EMS-FL in convergence performance. Taking the results in Fig. \ref{ex1_dflr}(a) as an example, with step-size $\eta^\mathrm{E} = 10^{-5}$, EMS-FL reaches an accuracy of 63\% in around 30 orbital cycles while the baseline scheme requires more than 60 orbital cycles, representing a 50\% reduction in training time. Moreover, Fig. \ref{ex1_compare} also compares the results obtained with different step-sizes for each scheme, indicating that a step-size of $\eta^\mathrm{E} = 1\times10^{-5}$ yields the fastest convergence. A smaller step-size leads to stable but slow accuracy improvement, while a large step-size brings in severe fluctuations that may cause an initial drop in accuracy.

\begin{figure}[t]
        \centering
        \subfigure[Fine-tuning process of the baseline scheme.]{\includegraphics[height=5cm, width=8cm]{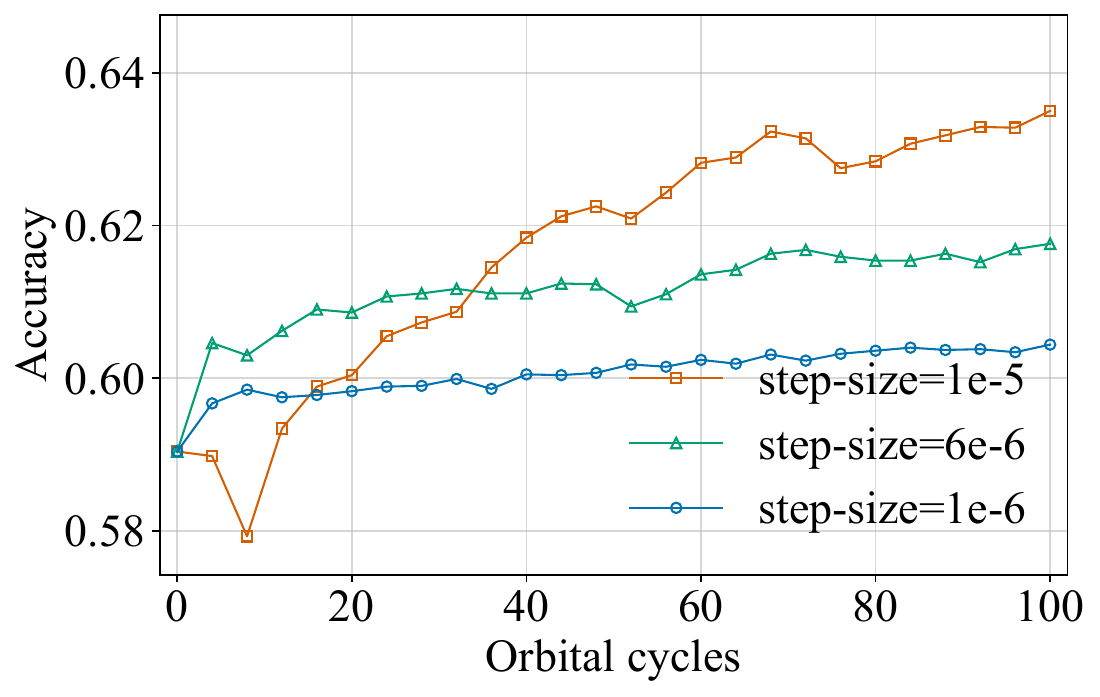}}
        \subfigure[Fine-tuning process of EMS-FL.]{\includegraphics[height=5cm, width=8cm]{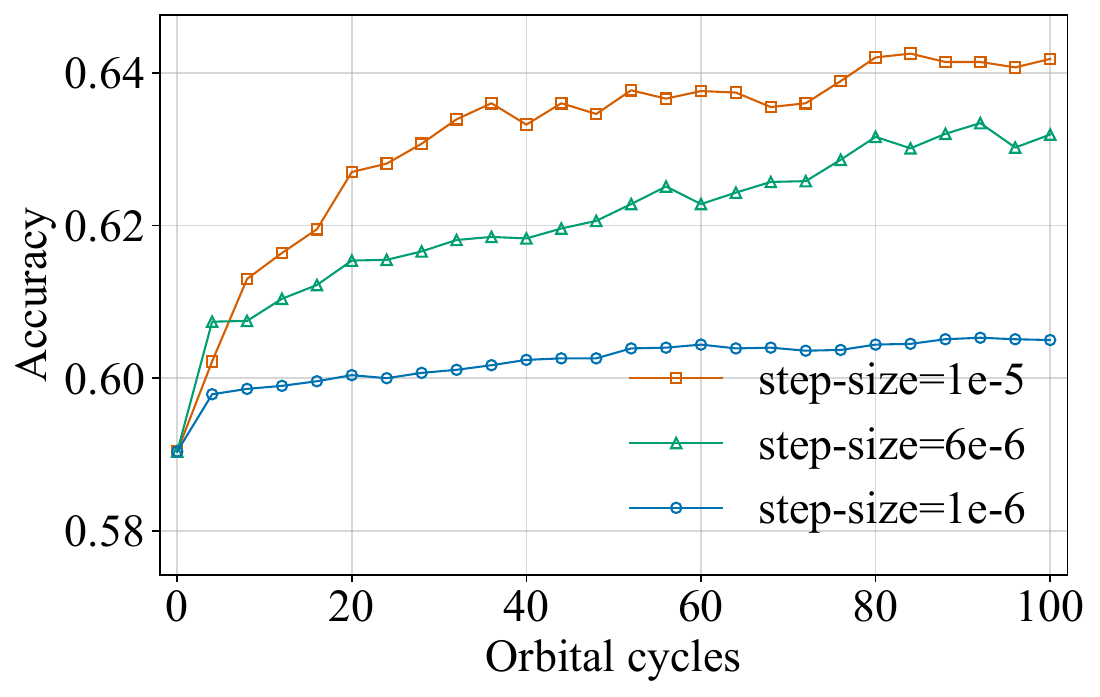}}
    \caption{Test accuracy under different step-sizes.}
    \label{ex1_compare}
\end{figure}

\begin{figure}[t]
	\centering
	\includegraphics[height=5cm, width=8cm]{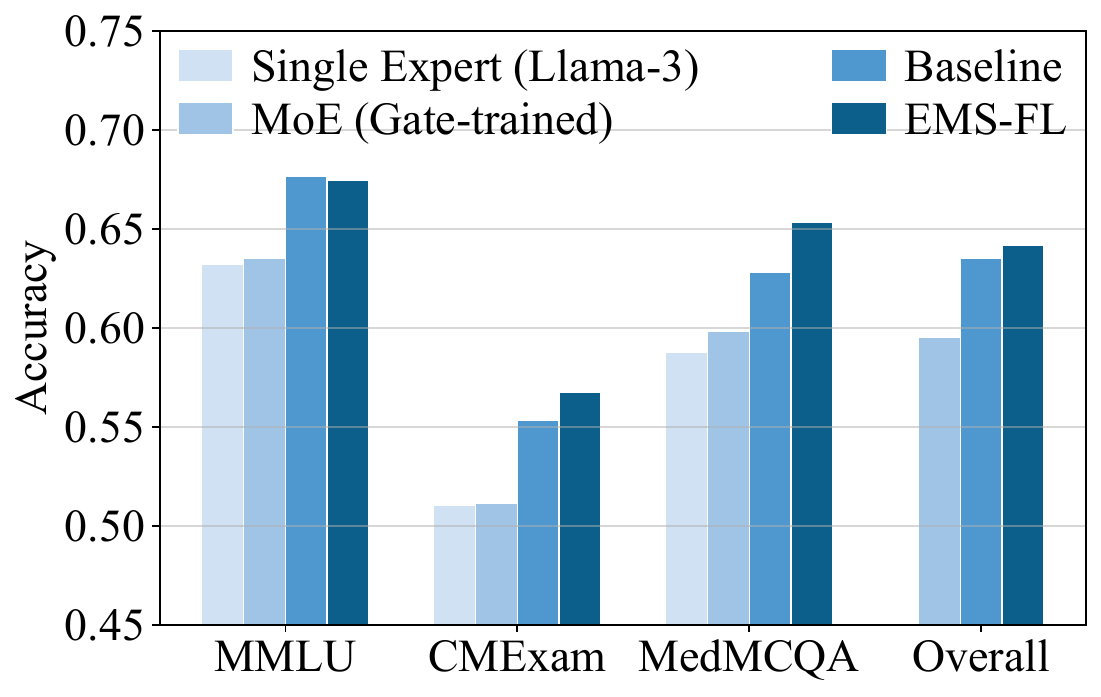}
	\caption{Test accuracy over different datasets before and after fine-tuning.}\label{ex1_inference}
\end{figure}

\begin{table*}[t]
\caption{Comparison of model capabilities before and after fine-tuning.}
\label{tab:gen-qualitative}
\centering
\justify
\footnotesize
\setlength{\tabcolsep}{8pt}
\renewcommand{\arraystretch}{1.18}
\sloppy 
\begin{tabularx}{\textwidth}{@{} L{5.0cm} L{5.0cm} L{6.7cm} @{}}
\toprule
\textbf{Prompt / Ground Truth} & \textbf{Before Fine-tuning} & \textbf{After Fine-tuning} \\
\midrule

\vspace{-0.4cm}
\justify{\textbf{Prompt:} List some vitamins that can only be supplied from animal source and explain the reason. } 
\vspace{-0.2cm}
\justify{\textbf{Ground Truth:} Vitamin B12}
&
\vspace{-0.4cm}
\justify{\underline{Vitamin A, Vitamin D}, and Vitamin B12 can only be supplied from animal source because they require fat for absorption, and fat is commonly found in animal products.}
&
\vspace{-0.4cm}
\begin{justify}
\underline{Vitamin B12} is an example of a vitamin that can only be supplied from animal source. Vitamin B12 is synthesized by bacteria in the gastrointestinal tract of animals and is stored in animal tissues such as liver and meat.\vspace{-0.2cm} 
\end{justify}
\\

\midrule

\vspace{-0.4cm}
\justify{\textbf{Prompt:} In chronic alcoholism, what is the rate-limiting component of alcohol metabolism, excluding enzymes? Explain the reason.}
\vspace{-0.2cm}
\justify{\textbf{Ground Truth:} NAD+ }
&
\vspace{-0.4cm}
\justify{The rate limiting component for alcohol metabolism excluding enzymes in chronic alcoholism is \underline{acetaldehyde dehydrogenase}.}
&
\vspace{-0.4cm}
\begin{justify}
In chronic alcoholism, the rate limiting component for alcohol metabolism excluding enzymes is the availability of \underline{NAD+}. This is because chronic alcohol consumption leads to the depletion of NAD+ due to its conversion to NADH during alcohol metabolism. \vspace{-0.2cm} 
\end{justify}
\\
\bottomrule
\end{tabularx}
\end{table*}

Moreover, as described in Sec. \ref{sec:enhanced}, the enhanced EMS-FL also help reduce the memory usage and communication overhead. To quantify this advantage, we record the peak GPU memory usage during the tuning process in Table \ref{tab:cost}, which reflects the memory demand for local devices. 
The results presented in Table \ref{tab:cost} corresponds a maximum input length of 384 tokens.
The results show that EMS-FL achieves a $49.5\%$ reduction in both communication and memory overhead, and only requires to upload a file with size less than 20 MB, which can be easily accomplished in the connected phase.

To demonstrate the capability enhancement achieved through fine-tuning, we compare the model outputs before and after tuning using representative question prompts relevant to the training dataset content, as shown in Table \ref{tab:gen-qualitative}. 
These results indicate that the MoE has effectively improved its inference capabilities, validating the experimental design.

In Fig. \ref{ex1_inference}, we also compare the test accuracy across different models over both the global and local datasets. For each local dataset, four models are involved in the comparison including 1) the most correlated pre-trained Llama-3 model, 2) the MoE model after gate pre-training, 3) the MoE model tuned with the baseline scheme, and 4) the MoE model tuned with EMS-FL. For the global dataset, we only involve the last three models. As shown in Fig. \ref{ex1_inference}, the model tuned with EMS-FL achieves the highest accuracy on CMExam, MedMCQA, and the global dataset, and performs comparably to the baseline scheme on MMLU, demonstrating its superiority. Moreover, tuned with either EMS-FL or the baseline scheme, the resultant MoE models provide apparently higher accuracy in all datasets when compared with the single Llama-3 model, justifying the effectiveness of the MoE architecture and STN-assisted FL. Last but not least, the MoE model after gate pre-training achieves similar accuracy to the single Llama-3 model, indicating that the gate has learned the correlations between experts and local datasets in pre-training and can route the data correctly.

\begin{figure}[t]
	\centering
	\includegraphics[height=5cm, width=8cm]{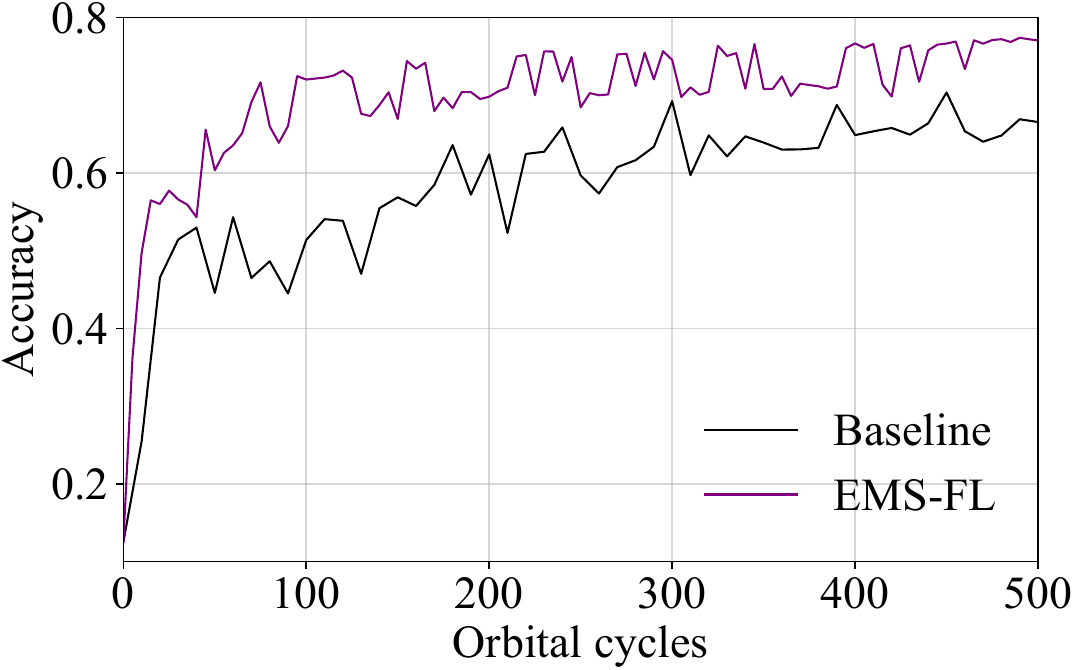}
	\caption{Test accuracy versus number of orbital cycles for EMS-FL and the baseline scheme with step-size $\eta^\mathrm{E} =8\times10^{-6}$.}\label{ex2_training}
\end{figure}

\subsection{Experiment 2: MoE Training with Top-2 Routing}\label{experiment2}
In the second experiment, we consider a training task to show the proposed EMS-FL can achieve favorable performance in both tuning and training scenarios. We employ Mixtral-8x7B-Instruct-v0.1, a 47B MoE model with Top-2 routing and 8 experts. Though LoRA is inapplicable to the training case, which may result in communication infeasibility, the simulation results still provide valuable insights for future research on mini-MoE \cite{Komatsuzaki2023SparseUpcycling} or MoE models with smaller sizes via knowledge distillation or model compression. Four device clusters are involved in this experiment, and four datasets are selected from the GLUE benchmark \cite{wang2019glue} as the corresponding local datasets. The four datasets  exhibit distinct linguistic characteristics: 1) sentence-level judgment (SLJ), 2) semantic similarity judgment (SSJ), 3) natural language inference (NLI), and 4) question-based natural language inference (QA-NLI). 
The model parameters are randomly initialized and trained with a step-size $\eta^\mathrm{E} \!=\! \eta^\mathrm{U} \!=\! 8\!\times\! 10^{-6}$. The experts are assigned to the four device clusters and trained asynchronously. The gate and the shared backbone are trained together by the device cluster connected to the satellite in each training round.
For the training details, we treat all four GLUE subsets as sentence-pair classification tasks, replacing the language modeling head with a classification head and training with cross-entropy loss.

\begin{figure}[t]
	\centering
	\includegraphics[height=5cm, width=8cm]{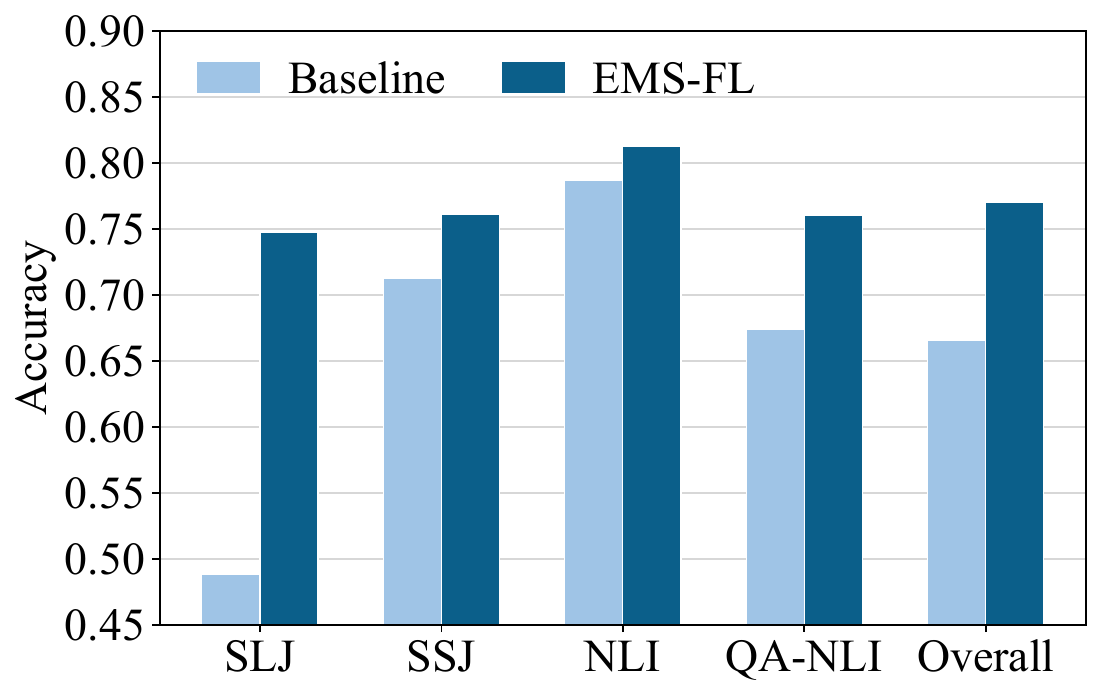}
	\caption{Test accuracy over different datasets after training.}\label{Ex2_inference}
\end{figure}

The comparison of convergence performance between EMS-FL and the baseline scheme is presented in Fig. \ref{ex2_training}. The accuracy is evaluated on the global dataset across $500$ orbital cycles. It is apparent that EMS-FL achieves a more rapid accuracy improvement and also a much faster convergence compared with the baseline scheme. Specifically, the proposed EMS-FL reaches the accuracy $70\%$ within 400 training rounds, while it takes over 1600 training rounds for the baseline scheme, representing a $75\%$ reduction on the training time.

We also compare the test accuracy of EMS-FL and the baseline scheme using both the global and local datasets, as shown in Fig. \ref{Ex2_inference}, where the proposed EMS-FL achieves notably higher accuracy over all datasets. Although such a training experiment may be infeasible in STN scenarios without LoRA, this suggests the potential of EMS-FL for training smaller MoE models, such as mini-MoE \cite{Komatsuzaki2023SparseUpcycling}, in future applications.

\section{Conclusion}\label{sec:conclusion}
In this work, we propose EMS-FL, an STN-assisted federated learning method that enables efficient MoE tuning in distributed scenarios by innovative designs of expert-driven model splitting and asynchronous local training. The superiority of EMS-FL in lightweight training, rapid convergence, and high accuracy is demonstrated through both theoretical analysis and comprehensive experiments. This work shows the possibility and provides a practical solution to distributed large model tuning by exploiting the assistance of STN. We believe it also sheds light on future designs of efficient model tuning and inference, potentially enhancing applications such as intelligent driving, AI agents, and embodied intelligence.

\section{Appendix}

\subsection{Proof of Theorem \ref{thm:emsfl}}\label{proof-theo1}
For any global training round $t$, after $C$ cluster contacts (from $t$ to $t+C$), we have
\begin{align}\label{eq:loss-decomp-orbit}
    & F(\boldsymbol e_{t+C},\boldsymbol u_{t+C})  - F(\boldsymbol e_t,\boldsymbol u_t)   \nonumber \\
    =& \underbrace{F(\!\boldsymbol e_{t+C},\boldsymbol u_t\!)\! -\! F(\!\boldsymbol e_t,\boldsymbol u_t\!)
    }_{\text{expert update}}
    \!+\!  \underbrace{ F(\!\boldsymbol e_{t+C},\boldsymbol u_{t+C}\!) \!-\! F(\boldsymbol e_{t+C},\boldsymbol u_t)}_{\text{gate update}}.
\end{align}
We first focus on the expert update part. Since ${\boldsymbol e}_t = [ \tilde{\boldsymbol e}_{t,1}^T,       \tilde{\boldsymbol e}_{t,2}^T,       \ldots,       \tilde{\boldsymbol e}_{t,C}^T    ]^T$, we have the following decomposition using Assumption~\ref{ass:smooth}
\begin{align}\label{eq-loss-decomp-term1}
    \!\!\!&F({\boldsymbol e}_{t+C},\boldsymbol u_t) - F({\boldsymbol e}_t,\boldsymbol u_t)
    \nonumber \\
    \!\!\!\le&  \langle \nabla_{{\boldsymbol e}} F({\boldsymbol e}_t,\boldsymbol u_t),
                {\boldsymbol e}_{t+C}-{\boldsymbol e}_t\rangle + \frac{L}{2}\|{\boldsymbol e}_{t+C}-{\boldsymbol e}_t\|^2  \nonumber \\
    \!\!\!=&\!   \sum_{c=1}^C\ [  \langle \nabla_{\tilde{\boldsymbol e}_c} \!F({\boldsymbol e}_t,\boldsymbol u_t), \tilde{\boldsymbol e}_{t+C,c}\!\!-\!\tilde{\boldsymbol e}_{t,c}\rangle
     \!+\! \frac{L}{2}\!\|\tilde{\boldsymbol e}_{t+C,c}\!\!-\!\tilde{\boldsymbol e}_{t,c}\|^2 ].
\end{align}
For the first term, its expectation satisfies
\begin{align}\label{t-to-tprime-gap}
    &\mathbb{E}\langle    \nabla_{\tilde{\boldsymbol e}_{c}}F({\boldsymbol e}_t,\boldsymbol u_t),
       \tilde{\boldsymbol e}_{t+C,\,c}-\tilde{\boldsymbol e}_{t,c}  \rangle    \nonumber \\
    =& \mathbb{E}\langle    \nabla_{\tilde{\boldsymbol e}_{c}}F({\boldsymbol e}_{t^\prime},\boldsymbol u_{t^\prime}),  \tilde{\boldsymbol e}_{{t^\prime}+C,c}-\tilde{\boldsymbol e}_{{t^\prime},c}  \rangle \nonumber\\
    &+ \mathbb{E}\langle    \nabla_{\tilde{\boldsymbol e}_{c}}F({\boldsymbol e}_t,\boldsymbol u_t) - \nabla_{\tilde{\boldsymbol e}_{c}}F({\boldsymbol e}_{t^\prime},\boldsymbol u_{t^\prime}),  \tilde{\boldsymbol e}_{{t^\prime}+C,c}-\tilde{\boldsymbol e}_{{t^\prime},c}  \rangle \nonumber\\
    \le& \mathbb{E}\langle    \nabla_{\tilde{\boldsymbol e}_{c}}F({\boldsymbol e}_{t^\prime},\boldsymbol u_{t^\prime}),  \tilde{\boldsymbol e}_{{t^\prime}+C,c}-\tilde{\boldsymbol e}_{{t^\prime},c}  \rangle \nonumber\\
    &+ L\| ({\boldsymbol e}_t, {\boldsymbol u_t}) - ({\boldsymbol e}_{t^\prime}, {\boldsymbol u}_{t^\prime}) \| \|\tilde{\boldsymbol e}_{{t^\prime}+C,c}-\tilde{\boldsymbol e}_{{t^\prime},c}\| \nonumber \\
    \le& \mathbb{E}\langle    \nabla_{\tilde{\boldsymbol e}_{c}}F({\boldsymbol e}_{t^\prime},\boldsymbol u_{t^\prime}),  \tilde{\boldsymbol e}_{{t^\prime}+C,c}-\tilde{\boldsymbol e}_{{t^\prime},c}  \rangle \nonumber\\
    &+ \!(\eta^{\mathrm E} \gamma)^2 \!LC^3(G_E^2 \!+\! \sigma_E^2) \!+\! \frac{1}{2} (\eta^{\mathrm U})^2\!LC^2(G_U^2 \!+\! \sigma_U^2),
\end{align}
with $C\ge2$, where ${t^\prime}$ is the last connected round for each cluster $c$ where $\tilde{\boldsymbol e}_c$ in the global model is kept frozen for $C$ rounds since then.
We further decompose $\mathbb{E}\langle    \nabla_{\tilde{\boldsymbol e}_{c}}F({\boldsymbol e}_{t^\prime},\boldsymbol u_{t^\prime}),  \tilde{\boldsymbol e}_{{t^\prime}+C,c}-\tilde{\boldsymbol e}_{{t^\prime},c}  \rangle$ as
\begin{align}\label{eq:ex_step2}
    &\mathbb{E}\langle    \nabla_{\tilde{\boldsymbol e}_c}F({\boldsymbol e}_{t^\prime},\boldsymbol u_{t^\prime}),
       \tilde{\boldsymbol e}_{t^\prime+C,c}-\tilde{\boldsymbol e}_{t^\prime,c}  \rangle    \nonumber \\
    =& \frac{1}{J}\sum_{\tau=t^\prime}^{t^\prime+C-1} \sum_{j\in\mathcal U_c}  \mathbb{E}\langle
       \nabla_{\tilde{\boldsymbol e}_c}F({\boldsymbol e}_{t^\prime},\boldsymbol u_{t^\prime}),
       -\eta^{\mathrm E} \gamma\mathbf g^{\mathrm E}_{\tau,c,j}  \rangle,
\end{align}
The multiplication can be rewritten as
\begin{align}
    &\mathbb{E}\langle \nabla_{\tilde{\boldsymbol e}_c} F({\boldsymbol e}_{t^\prime},\boldsymbol u_{t^\prime}), -\eta^{\mathrm E} \gamma\mathbf g^{\mathrm E}_{\tau,c,j}  \rangle \nonumber\\
    =& \mathbb{E}[\nabla_{\tilde{\boldsymbol e}_c} F({\tilde{\boldsymbol e}}_{\tau,c,j},\boldsymbol u_{t^\prime})({\tilde{\boldsymbol e}}_{\tau+1,c,j} - {\tilde{\boldsymbol e}}_{\tau,c,j}) \nonumber\\
    &+ \left(\nabla_{\tilde{\boldsymbol e}_c} F({\boldsymbol e}_{t^\prime},\boldsymbol u_{t^\prime})   -   \nabla_{\tilde{\boldsymbol e}_c} F({\tilde{\boldsymbol e}}_{\tau,c,j},\boldsymbol u_{t^\prime})\right)({\tilde{\boldsymbol e}}_{\tau+1,c,j} - {\tilde{\boldsymbol e}}_{\tau,c,j})]\nonumber\\
    \le& -\! \eta^{\mathrm E}\!\gamma\mathbb{E}[\nabla_{\tilde{\boldsymbol e}_c}^2
    \!F(\!{\tilde{\boldsymbol e}}_{\tau,c,j},\boldsymbol u_{t^\prime}\!) \!+\!\! L\|{\tilde{\boldsymbol e}}_{t^\prime,c} \!\!-\!  {\tilde{\boldsymbol e}}_{\tau,c,j}\| \|{\tilde{\boldsymbol e}}_{\tau\!+\!1,c,j} \!\!-\! {\tilde{\boldsymbol e}}_{\tau,c,j} \|  ] \nonumber\\
    \le& - \eta^{\mathrm E} \gamma\mathbb{E}[\nabla_{\tilde{\boldsymbol e}_c}^2 F({\tilde{\boldsymbol e}}_{\tau,c,j},\boldsymbol u_{t^\prime})] + (\eta^{\mathrm E} \gamma)^2 L C (G_E^2 + \sigma_E^2).  \nonumber
\end{align}
Since
\begin{align}
   &\| \nabla_{\tilde{\boldsymbol e}_c} \!F({\tilde{\boldsymbol e}}_{\tau,c,j},\boldsymbol u_{t^\prime}) \| \nonumber\\
   \ge& \| \nabla_{\tilde{\boldsymbol e}_c}\! F({{\boldsymbol e}}_{t},\boldsymbol u_{t})  \| \!-\! L\|({\tilde{\boldsymbol e}}_{t,c},\boldsymbol u_{t}) - ({\tilde{\boldsymbol e}}_{\tau,c,j}, \boldsymbol u_{t^\prime})\|,
\end{align}
which gives
\begin{align}
   \!\!&\mathbb{E}[\| \nabla_{\tilde{\boldsymbol e}_c}^2 F({\tilde{\boldsymbol e}}_{\tau,c,j},\boldsymbol u_{t^\prime}) \|] \nonumber\\
   \ge & \mathbb{E}[\| \nabla_{\tilde{\boldsymbol e}_c}^2 F({{\boldsymbol e}}_{t},\boldsymbol u_{t})  \| - 2LG_E \|({\tilde{\boldsymbol e}}_{t,c}, \boldsymbol u_{t}) \!-\!  ({\tilde{\boldsymbol e}}_{\tau,c,j}, \boldsymbol u_{t^\prime})\|] \nonumber\\
   \!\!\ge& \mathbb{E}[\| \nabla_{\tilde{\boldsymbol e}_c}^2\! F(\!{{\boldsymbol e}}_{t},\boldsymbol u_{t}\!)  \|] \!-\! 2LC\eta^{\mathrm E} \gamma[  G_{E}^2 \!+\! \sigma_{E}^2 \!+\! (\!\frac{\eta^{\mathrm U}}{\eta^{\mathrm E} \gamma}\!)^2 (\!G_{U}^2 \!+\! \sigma_{U}^2\!)],\nonumber
\end{align}
where we use the inequality
\begin{align}
E[\|X\|]\!=\!\sqrt{\!(E[\|X\| \!\cdot\! 1])^2} 
    \!\le\! \sqrt{\!E[\|X\|^2] \!\cdot\! E\left[1^2\right]}
    \!=\! \sqrt{\!E[\|X\|^2]}, \nonumber
\end{align}
for an arbitrary variable $X$ according to the Cauchy-Schwarz inequality.
Therefore we have
\begin{align}
&\mathbb{E}\langle \nabla_{\tilde{\boldsymbol e}_c} F({\boldsymbol e}_{t^\prime},\boldsymbol u_{t^\prime}), -\eta^{\mathrm E} \gamma\mathbf g^{\mathrm E}_{\tau,c,j}  \rangle \nonumber\\
\le&  - \!\eta^{\mathrm E} \gamma\mathbb{E}[\nabla_{\tilde{\boldsymbol e}_c}^2 F({{\boldsymbol e}}_{t},\boldsymbol u_{t})] \nonumber \\
&+ 3(\eta^{\mathrm E} \gamma)^2 LC (G_E^2 + \sigma_E^2) + 2(\eta^{\mathrm U})^2 LC (G_U^2 + \sigma_U^2).
\end{align}
Then the result in  \eqref{eq:ex_step2} can be bounded as
\begin{align}\label{first-term-result}
&\mathbb{E}\langle    \nabla_{\tilde{\boldsymbol e}_c}F({\boldsymbol e}_{t^\prime},\boldsymbol u_{t^\prime}), \tilde{\boldsymbol e}_{t^\prime+C,c}-\tilde{\boldsymbol e}_{t^\prime,c}  \rangle    \nonumber \\
\le& - \eta^{\mathrm E} \gamma C  \mathbb{E}[\nabla_{\tilde{\boldsymbol e}_c}^2 F({{\boldsymbol e}}_{t},\boldsymbol u_{t})] \nonumber\\
&+ \!3(\eta^{\mathrm E} \gamma)^2 C^2 \!L (G_E^2 \!+\! \sigma_E^2) \!+\! 2(\eta^{\mathrm U})^2\! LC^2 (G_U^2 \!+\! \sigma_U^2).
\end{align}
Moreover, the second term in \eqref{eq-loss-decomp-term1} is directly bounded by
\begin{align}\label{second-term-result}
    \frac{L}{2}\mathbb{E}[\|\tilde{\boldsymbol e}_{t+C,c}-\tilde{\boldsymbol e}_{t,c}\|^2] \le \frac{1}{2}(\eta^{\mathrm E} \gamma)^2 C^2 L (G_E^2 + \sigma_E^2).
\end{align}
Plugging \eqref{t-to-tprime-gap}, \eqref{first-term-result} and \eqref{second-term-result} into \eqref{eq-loss-decomp-term1} gives
\begin{align}\label{upper-bound-term1}
    &\mathbb{E}[F({\boldsymbol e}_{t+C},\boldsymbol u_t) - F({\boldsymbol e}_t,\boldsymbol u_t)]\nonumber\\
    \le& - \eta^{\mathrm E} \gamma C  \mathbb{E}[\nabla_{{\boldsymbol e}}^2 F({{\boldsymbol e}}_t,\boldsymbol u_t)] \nonumber\\
    &+ \!3 (\eta^{\mathrm E} \gamma)^2 C^4\! L (G_E^2 \!+\! \sigma_E^2) \!+\! 3(\eta^{\mathrm U})^2\!LC^3(G_U^2 \!+\! \sigma_U^2),
\end{align}
for $C\ge2$, which is the upper bound of the expert update part in \eqref{eq:loss-decomp-orbit}.

Next, for the gate update part in \eqref{eq:loss-decomp-orbit}, we first have
\begin{align}\label{eq:loss-decomp-term2}
    &F({\boldsymbol e}_{t+C},\boldsymbol u_{t+C})     - F({\boldsymbol e}_{t+C},\boldsymbol u_t)    \nonumber \\
    &=    \sum_{\tau=t}^{t+C-1}    [       F({\boldsymbol e}_{t+C},\boldsymbol u_{\tau+1})     - F({\boldsymbol e}_{t+C},\boldsymbol u_{\tau})    ].
\end{align}
Using $L$-smoothness, we have
\begin{align}\label{eq:gate_step1}
    &F({\boldsymbol e}_{t+C}, \boldsymbol u_{\tau+1}) - F({\boldsymbol e}_{t+C}, \boldsymbol u_\tau)
    \nonumber \\
    \le&    \langle       \nabla_{\boldsymbol u} F({\boldsymbol e}_{t+C}, \boldsymbol u_\tau),\,
       \boldsymbol u_{\tau+1} - \boldsymbol u_\tau
    \rangle     + \frac{L}{2} \|\boldsymbol u_{\tau+1} - \boldsymbol u_\tau\|^2.
\end{align}
Let $c(\tau)$ denote the device cluster in the connected phase in round $\tau$, the first term in \eqref{eq:gate_step1} can be written as
\begin{align}\label{eq:gate_step3}
    & \mathbb{E}\langle         \nabla_{\boldsymbol u} F({\boldsymbol e}_{t+C}, \boldsymbol u_\tau),
        \boldsymbol u_{\tau+1} - \boldsymbol u_\tau    \rangle    \nonumber \\
    =&    \frac{1}{J}   \sum_{j\in\mathcal U_{c(\tau)}} \mathbb{E}  
    \langle     \nabla_{\boldsymbol u} F({\boldsymbol e}_{t+C}, \boldsymbol u_\tau),
        -\eta^{\mathrm U}\mathbf g^{\mathrm U}_{\tau,c(\tau),j}    \rangle, \nonumber\\
    =& \frac{1}{J}   \sum_{j\in\mathcal U_{c(\tau)}}    \mathbb{E}[\langle\nabla_{\boldsymbol u} F({\boldsymbol e}_{\tau}, \boldsymbol u_\tau),    -\eta^{\mathrm U}\mathbf g^{\mathrm U}_{\tau,c(\tau),j}\rangle \nonumber\\
    &+  \langle \nabla_{\boldsymbol u} F({\boldsymbol e}_{t+C}, \boldsymbol u_\tau)-\nabla_{\boldsymbol u} F({\boldsymbol e}_{\tau}, \boldsymbol u_\tau), -\eta^{\mathrm U}\mathbf g^{\mathrm U}_{\tau,c(\tau),j} \rangle]\nonumber\\
    \le& \frac{1}{J}   \sum_{j\in\mathcal U_{c(\tau)}}\mathbb{E}[-\eta^{\mathrm U}\nabla_{\boldsymbol u}^2 F({\boldsymbol e}_{\tau}, \boldsymbol u_\tau) \nonumber \\
    &+ L\|{\boldsymbol e}_{t+C} - {\boldsymbol e}_{\tau} \|\|\boldsymbol u_{\tau+1,c(\tau),j} - \boldsymbol u_\tau \| ]\nonumber \\
    \le& \frac{1}{J}   \sum_{j\in\mathcal U_{c(\tau)}}\mathbb{E}[-\eta^{\mathrm U}\nabla_{\boldsymbol u}^2 F({\boldsymbol e}_{\tau}, \boldsymbol u_\tau) \nonumber \\
    &+ \frac{L}{2}(\|{\boldsymbol e}_{t+C} - {\boldsymbol e}_{t} \|^2 + \|\boldsymbol u_{\tau+1,c(\tau),j} - \boldsymbol u_\tau \|^2) ]\nonumber \\
    \le& -\eta^{\mathrm U}\mathbb{E}[\nabla_{\boldsymbol u}^2 F({\boldsymbol e}_{\tau}, \boldsymbol u_\tau)] \nonumber \\
    &+ \frac{L}{2} C^3 (\eta^\mathrm{E}\gamma)^2 (G_E^2+\sigma_E^2) + \frac{L}{2}(\eta^\mathrm{U})^2(G_U^2+\sigma_U^2)
\end{align}
Moreover, since
\begin{align}
    \|\nabla_{\boldsymbol u} F({\boldsymbol e}_{\tau}, \boldsymbol u_\tau)\| \ge \|\nabla_{\boldsymbol u} F({\boldsymbol e}_{t}, \boldsymbol u_t)\| \!-\! L\|({\boldsymbol e}_{\tau},{\boldsymbol u}_{\tau}) - ({\boldsymbol e}_{t}, {\boldsymbol u}_{t}) \|,\nonumber
\end{align}
we further have
\begin{align}
    &\mathbb{E}[\nabla_{\boldsymbol u}^2 F({\boldsymbol e}_{\tau}, \boldsymbol u_\tau)] \nonumber\\
    \ge& \mathbb{E}[\nabla_{\boldsymbol u}^2 F({\boldsymbol e}_{t}, \boldsymbol u_t) - 2L G_U\|({\boldsymbol e}_{\tau},{\boldsymbol u}_{\tau}) - ({\boldsymbol e}_{t}, {\boldsymbol u}_{t}) \|  ]\nonumber\\
    \ge& \mathbb{E}[\nabla_{\boldsymbol u}^2 F({\boldsymbol e}_{t}, \boldsymbol u_t)] \!-\! 2\eta^\mathrm{U}LC[G_U^2\!+\!\sigma_U^2\!+\!(\!\frac{C\eta^\mathrm{E}\gamma}{\eta^\mathrm{U}}\!)^2(G_E^2\!+\!\sigma_E^2)].\nonumber
\end{align}
And the result in \eqref{eq:gate_step3} can be rewritten as
\begin{align}\label{eq:gate_step3-final}
    & \mathbb{E}\langle         \nabla_{\boldsymbol u} F({\boldsymbol e}_{t+C}, \boldsymbol u_\tau), 
        \boldsymbol u_{\tau+1} - \boldsymbol u_\tau    \rangle    \nonumber \\
    \le& -\eta^\mathrm{U} \mathbb{E}[\nabla_{\boldsymbol u}^2 F({\boldsymbol e}_{t}, \boldsymbol u_t)] \nonumber \\
    &+ \!2(\eta^\mathrm{E}\gamma)^2\! LC^3(G_E^2\!+\!\sigma_E^2)\!+\!\frac{5}{2}(\eta^\mathrm{U})^2\!LC(G_U^2\!+\!\sigma_U^2),
\end{align}
with $C\ge2$. Also, the second term can be bounded by
\begin{align}\label{eq:gate_step4}
    \frac{L}{2} \mathbb{E}\|\boldsymbol u_{\tau+1} - \boldsymbol u_\tau\|^2] \le \frac{L}{2}(\eta^\mathrm{U})^2(G_U^2+\sigma_U^2).
\end{align}
Therefore, by plugging \eqref{eq:gate_step3-final} and \eqref{eq:gate_step4} into \eqref{eq:loss-decomp-term2}, we have
\begin{align}\label{upper-bound-term2}
    &\mathbb{E}[F({\boldsymbol e}_{t+C},\boldsymbol u_{t+C})  - F({\boldsymbol e}_{t+C},\boldsymbol u_t)] \nonumber\\
    \le& -\eta^{\mathrm U}C\mathbb{E}[\nabla_{\boldsymbol u}^2 F({\boldsymbol e}_{t}, \boldsymbol u_t)] \nonumber\\
    &+\!2(\eta^\mathrm{E}\gamma)^2\! LC^4(G_E^2\!+\!\sigma_E^2)\!+\!3(\eta^\mathrm{U})^2LC^2(G_U^2\!+\!\sigma_U^2).
\end{align}

Combining the results in \eqref{upper-bound-term1} and \eqref{upper-bound-term2}, we finally obtain
\begin{align}
     &F(\boldsymbol e_{t+C},\boldsymbol u_{t+C})  - F(\boldsymbol e_t,\boldsymbol u_t) \nonumber\\
     \le& - \eta^{\mathrm E} \gamma C  \mathbb{E}[\nabla_{{\boldsymbol e}}^2 F({{\boldsymbol e}}_t,\boldsymbol u_t)] -\eta^{\mathrm U}C\mathbb{E}[\nabla_{\boldsymbol u}^2 F({\boldsymbol e}_{t}, \boldsymbol u_t)]\nonumber\\
     &+ \!5 (\eta^{\mathrm E} \gamma)^2 C^4\! L (G_E^2 \!+\! \sigma_E^2) \!+\! 5(\eta^{\mathrm U})^2\!LC^3(G_U^2 \!+\! \sigma_U^2),
\end{align}
for $C\ge2$. By summing the above results at rounds $t \in \{0, C, \ldots, (T-1)C\}$ with $\eta^{\mathrm E} \le \frac{\eta^{\mathrm U}}{\gamma} \le \frac{1}{\sqrt{T}}$, we further obtain
\begin{align}
    &\frac{1}{T}  \sum_{t=0,\,C,\,\ldots,\,(T-1)C}  \mathbb E\|\nabla F(\boldsymbol\theta_t)\|^2 \nonumber \\
    \le& \frac{\mathbb E[F(\boldsymbol\theta_0) \!-\! F^\star]}{\gamma C\sqrt{T}}  \!+\! \frac{5L C^2\gamma[ C (G_E^2 \!+\! \sigma_E^2) \!+\! G_U^2\!+\!\sigma_U^2]}{\sqrt{T}}.
\end{align}
Note that similar results can still be obtained with $\eta^{\mathrm E} = q\eta^{\mathrm U} \le \frac{1}{\sqrt{T}}$ where $q \le \frac{1}{\gamma}$.

\subsection{Proof of Theorem~\ref{thm:baseline}}\label{proof-theo2}
For the baseline scheme, the decomposition in \eqref{eq:loss-decomp-orbit} still holds. Since all the experts update simultaneously, the expert update part can be rewritten as
\begin{align}\label{eq:baseline-loss-decomp-term1}
    \!\!\!&F({\boldsymbol e}_{t+C},\boldsymbol u_t) - F({\boldsymbol e}_t,\boldsymbol u_t)
    \nonumber \\
    \!\!\!\le&  \langle \nabla_{{\boldsymbol e}} F({\boldsymbol e}_t,\boldsymbol u_t),
                {\boldsymbol e}_{t+C}-{\boldsymbol e}_t\rangle + \frac{L}{2}\|{\boldsymbol e}_{t+C}-{\boldsymbol e}_t\|^2
\end{align}
For the first term, we have
\begin{align}
\label{baseline-first-term-tau-sum}
&\mathbb E\langle \nabla_{\boldsymbol e}F(\boldsymbol e_t,\boldsymbol u_t),\boldsymbol e_{t+C}-\boldsymbol e_t\rangle \nonumber\\
=&\sum_{\tau=t}^{t+C-1}\mathbb E\langle \nabla_{\boldsymbol e}F(\boldsymbol e_t,\boldsymbol u_t),\boldsymbol e_{\tau+1}-\boldsymbol e_\tau\rangle.
\end{align}
The multiplication can be rewritten as
\begin{align}\label{baseline-mult-bound-compact-fix}
&\mathbb E\langle \nabla_{\boldsymbol e}F(\boldsymbol e_t,\boldsymbol u_t),
\boldsymbol e_{\tau+1}-\boldsymbol e_\tau\rangle \nonumber\\
=&\mathbb E[
\langle \nabla_{\boldsymbol e}F(\boldsymbol e_\tau,\boldsymbol u_\tau),\boldsymbol e_{\tau+1}-\boldsymbol e_\tau\rangle \nonumber\\
&+\langle \nabla_{\boldsymbol e}F(\boldsymbol e_t,\boldsymbol u_t)-\nabla_{\boldsymbol e}F(\boldsymbol e_\tau,\boldsymbol u_\tau),\boldsymbol e_{\tau+1}-\boldsymbol e_\tau\rangle
]\nonumber\\
\le&
\mathbb E\langle \nabla_{\boldsymbol e}F(\boldsymbol e_\tau,\boldsymbol u_\tau),\boldsymbol e_{\tau+1}-\boldsymbol e_\tau\rangle \nonumber\\
&+L\,\mathbb E[\|({\boldsymbol e}_{\tau}, {\boldsymbol u}_{\tau})-({\boldsymbol e}_{t},{\boldsymbol u}_{t})\|\,\|\boldsymbol e_{\tau+1}-\boldsymbol e_\tau\|]\nonumber\\
\le&
-\eta^{\mathrm E}\mathbb E\|\nabla_{\boldsymbol e}^2F(\boldsymbol e_\tau,\boldsymbol u_\tau)\| + (\eta^{\mathrm U})^2LC^2(G_U^2+\sigma_U^2) \nonumber\\
&+ (\eta^{\mathrm E})^2L(C^2+\frac12)(G_E^2 + \zeta_E^2 + \sigma_E^2).
\end{align}
We also have
\begin{align}
    &\mathbb{E}[\nabla_{\boldsymbol e}^2 F({\boldsymbol e}_{\tau}, \boldsymbol u_\tau)] \nonumber\\
    \ge& \mathbb{E}[\nabla_{\boldsymbol e}^2 F({\boldsymbol e}_{t}, \boldsymbol u_t) - 2L \|\nabla_{\boldsymbol e} F({\boldsymbol e}_{t}, \boldsymbol u_t)\|\|({\boldsymbol e}_{\tau},{\boldsymbol u}_{\tau}) \!-\! ({\boldsymbol e}_{t}, {\boldsymbol u}_{t}) \|  ]\nonumber\\
    \ge& \mathbb{E}[\nabla_{\boldsymbol e}^2 F({\boldsymbol e}_{t}, \boldsymbol u_t)] \nonumber\\
    &- 2\eta^\mathrm{E}LC[G_E^2 + \zeta_E^2 + \sigma_E^2+(\!\frac{\eta^\mathrm{U}}{\eta^\mathrm{E}})^2(G_U^2+\sigma_U^2)].
\end{align}
The the result in \eqref{baseline-first-term-tau-sum} is then upper-bounded by
\begin{align}
\label{baseline-first-term-result}
&\mathbb E\langle \nabla_{\boldsymbol e}F(\boldsymbol e_t,\boldsymbol u_t),\boldsymbol e_{t+C}-\boldsymbol e_t\rangle \nonumber\\
\le&
-\eta^{\mathrm E}C\mathbb E\|\nabla_{\boldsymbol e}^2F(\boldsymbol e_t,\boldsymbol u_t)\| + 2(\eta^{\mathrm U})^2LC^3(G_U^2+\sigma_U^2) \nonumber\\
&+ (\eta^{\mathrm E})^2 L C(2C^2+\frac{1}{2}) (G_E^2 + \zeta_E^2 + \sigma_E^2).
\end{align}
Moreover, the second term in \eqref{eq:baseline-loss-decomp-term1} is directly bounded by
\begin{align}\label{eq:baseline-step-bound-e}
\frac{L}{2}\mathbb E\|\boldsymbol e_{t+C}-\boldsymbol e_t\|^2
\le \frac{L}{2}(\eta^{\mathrm E})^2C^2(G_E^2 + \zeta_E^2 + \sigma_E^2).
\end{align}
Plugging \eqref{baseline-first-term-result} and \eqref{eq:baseline-step-bound-e} into \eqref{eq:baseline-loss-decomp-term1} gives the upper bound of the expert update part as
\begin{align}\label{eq:baseline-term1-final}
&\mathbb E\!\left[F(\boldsymbol e_{t+C},\boldsymbol u_t)-F(\boldsymbol e_t,\boldsymbol u_t)\right] \nonumber\\
\le&
-\eta^{\mathrm E}C\,\mathbb E\|\nabla_{\boldsymbol e}F(\boldsymbol e_t,\boldsymbol u_t)\|^2 + 3(\eta^{\mathrm U})^2L\,C^{3}(G_U^2+\sigma_U^2)
\nonumber\\
&+ 3(\eta^{\mathrm E})^2L\,C^{3}
(G_E^2+\zeta_E^2+ \sigma_E^2),
\end{align}

For the gate update part in \eqref{eq:loss-decomp-orbit}, we have
\begin{align}\label{eq:baseline-loss-decomp-term2}
    \!\!\!&F({\boldsymbol e}_{t+C},\boldsymbol u_{t+C})     - F({\boldsymbol e}_{t+C},\boldsymbol u_t)    \nonumber \\
    \!\!\!=& \!\!\!  \sum_{\tau=t}^{t+C-1}    [ F({\boldsymbol e}_{t+C},\boldsymbol u_{\tau+1})     - F({\boldsymbol e}_{t+C},\boldsymbol u_{\tau})]\nonumber\\
    \!\!\!\le&  \!\!\!\sum_{\tau=t}^{t+C-1}  \!\!  [  \langle       \nabla_{\!\boldsymbol u} F({\boldsymbol e}_{t+C}, \boldsymbol u_\tau),\,
       \boldsymbol u_{\tau+1} \!\!-\! \boldsymbol u_\tau
    \rangle     \!+\! \frac{L}{2}\! \|\boldsymbol u_{\tau+1} \!\!-\! \boldsymbol u_\tau\|^2].
\end{align}
Let $c(\tau)$ denote the device cluster in the connected phase in round $\tau$, the first term can be written as
\begin{align}\label{eq:baseline-gate_step2}
    & \mathbb{E}\langle         \nabla_{\boldsymbol u} F({\boldsymbol e}_{t+C}, \boldsymbol u_\tau),
        \boldsymbol u_{\tau+1} - \boldsymbol u_\tau    \rangle    \nonumber \\
    =&    \frac{1}{J}   \sum_{j\in\mathcal U_{c(\tau)}} \mathbb{E}  
    \langle     \nabla_{\boldsymbol u} F({\boldsymbol e}_{t+C}, \boldsymbol u_\tau),
        -\eta^{\mathrm U}\mathbf g^{\mathrm U}_{\tau,c(\tau),j}    \rangle, \nonumber\\
    =& \frac{1}{J}   \sum_{j\in\mathcal U_{c(\tau)}}    \mathbb{E}[\langle\nabla_{\boldsymbol u} F({\boldsymbol e}_{\tau}, \boldsymbol u_\tau),    -\eta^{\mathrm U}\mathbf g^{\mathrm U}_{\tau,c(\tau),j}\rangle \nonumber\\
    &+  \langle \nabla_{\boldsymbol u} F({\boldsymbol e}_{t+C}, \boldsymbol u_\tau)-\nabla_{\boldsymbol u} F({\boldsymbol e}_{\tau}, \boldsymbol u_\tau), -\eta^{\mathrm U}\mathbf g^{\mathrm U}_{\tau,c(\tau),j} \rangle]\nonumber\\
    \le& \frac{1}{J}   \sum_{j\in\mathcal U_{c(\tau)}}\mathbb{E}[-\eta^{\mathrm U}\nabla_{\boldsymbol u}^2 F({\boldsymbol e}_{\tau}, \boldsymbol u_\tau) \nonumber \\
    &+ L\|{\boldsymbol e}_{t+C} - {\boldsymbol e}_{\tau} \|\|\boldsymbol u_{\tau+1,c(\tau),j} - \boldsymbol u_\tau \| ]\nonumber \\
    \le& -\eta^{\mathrm U}\mathbb{E}[\nabla_{\boldsymbol u}^2 F({\boldsymbol e}_{\tau}, \boldsymbol u_\tau)] \nonumber \\
    &+ \frac{L}{2} C^2 (\eta^\mathrm{E})^2 (G_E^2 \!+\! \zeta_E^2 \!+\! \sigma_E^2) \!+\! \frac{L}{2}(\eta^\mathrm{U})^2(G_U^2\!+\!\sigma_U^2)
\end{align}
We also have
\begin{align}
    &\mathbb{E}[\nabla_{\boldsymbol u}^2 F({\boldsymbol e}_{\tau}, \boldsymbol u_\tau)] \nonumber\\
    \ge& \mathbb{E}[\nabla_{\boldsymbol u}^2 F({\boldsymbol e}_{t}, \boldsymbol u_t) - 2L \|\nabla_{\boldsymbol u} F({\boldsymbol e}_{t}, \boldsymbol u_t)\|\|{\boldsymbol e}_{\tau} \!-\! {\boldsymbol e}_{t} \!+\! {\boldsymbol u}_{\tau} \!-\! {\boldsymbol u}_{t} \|  ]\nonumber\\
    \ge& \mathbb{E}[\nabla_{\boldsymbol u}^2\! F({\boldsymbol e}_{t}, \boldsymbol u_t)] \!-\! 2\eta^{\!\mathrm{U}}\! LC[(\!\frac{\eta^\mathrm{E}}{\eta^\mathrm{U}}\!)^2(\!G_E^2 \!+\! \zeta_E^2 \!+\! \sigma_E^2\!)\!+\!G_U^2\!+\!\sigma_U^2].\nonumber
\end{align}
Moreover, the second term in \eqref{eq:baseline-loss-decomp-term2} can be bounded by
\begin{align}\label{eq:baseline-gate_step3}
    \frac{L}{2} \mathbb{E}\|\boldsymbol u_{\tau+1} - \boldsymbol u_\tau\|^2] \le \frac{L}{2}(\eta^\mathrm{U})^2(G_U^2+\sigma_U^2).
\end{align}
Therefore, the gate update part can be bounded as
\begin{align}\label{eq:baseline-term2-final}
&\mathbb E\!\left[F(\boldsymbol e_{t+C},\boldsymbol u_{t+C})-F(\boldsymbol e_{t+C},\boldsymbol u_t)\right] \nonumber\\
\le&
-\eta^{\mathrm U}C\,\mathbb E\!\left[\nabla_{\boldsymbol u}^2F(\boldsymbol e_t,\boldsymbol u_t)\right]
+3 (\eta^{\mathrm U})^2LC^2(G_U^2+\sigma_U^2) \nonumber\\
&+2(\eta^{\mathrm E})^2LC^{3}(G_E^2+\zeta_E^2+\sigma_E^2).
\end{align}

Combining the results in  \eqref{eq:baseline-term1-final} and \eqref{eq:baseline-term2-final}, we finally obtain
\begin{align}
&\mathbb E\!\left[F(\boldsymbol e_{t+C},\boldsymbol u_{t+C})-F(\boldsymbol e_t,\boldsymbol u_t)\right] \nonumber\\
\le\;&
-\eta^{\mathrm E}C\,\mathbb E\|\nabla_{\boldsymbol e}F(\boldsymbol e_t,\boldsymbol u_t)\|^2
-\eta^{\mathrm U}C\,\mathbb E\!\left[\nabla_{\boldsymbol u}^2F(\boldsymbol e_t,\boldsymbol u_t)\right]
\nonumber\\
&+ 5(\eta^{\mathrm E})^2L\,C^{3}
(G_E^2+\zeta_E^2+\sigma_E^2) \nonumber\\
&+ 5(\eta^{\mathrm U})^2L\,C^{3}(G_U^2+\sigma_U^2).
\end{align}
for $C\ge2$. By summing the above results at rounds $t \in \{0, C, \ldots, (T-1)C\}$ with $\eta^{\mathrm E} = \frac{\eta^{\mathrm U}}{\gamma} = \frac{1}{\sqrt{T}}$, the result in \eqref{eq:emsfl-theorem} can be directly obtained.

\bibliographystyle{IEEEtran}
\bibliography{refs}

\begin{thebibliography}{10}
\providecommand{\url}[1]{#1}
\csname url@samestyle\endcsname
\providecommand{\newblock}{\relax}
\providecommand{\bibinfo}[2]{#2}
\providecommand{\BIBentrySTDinterwordspacing}{\spaceskip=0pt\relax}
\providecommand{\BIBentryALTinterwordstretchfactor}{4}
\providecommand{\BIBentryALTinterwordspacing}{\spaceskip=\fontdimen2\font plus
\BIBentryALTinterwordstretchfactor\fontdimen3\font minus \fontdimen4\font\relax}
\providecommand{\BIBforeignlanguage}[2]{{%
\expandafter\ifx\csname l@#1\endcsname\relax
\typeout{** WARNING: IEEEtran.bst: No hyphenation pattern has been}%
\typeout{** loaded for the language `#1'. Using the pattern for}%
\typeout{** the default language instead.}%
\else
\language=\csname l@#1\endcsname
\fi
#2}}
\providecommand{\BIBdecl}{\relax}
\BIBdecl

\bibitem{Mao2017MEC}
Y.~Mao, C.~You, J.~Zhang, K.~Huang, and K.~B. Letaief, ``A survey on mobile edge computing: The communication perspective,'' \emph{IEEE Commun. Surveys Tuts.}, vol.~19, no.~4, pp. 2322--2358, 4th Quart 2017.

\bibitem{Letaief2022EdgeAI6G}
K.~B. Letaief, Y.~Shi, J.~Lu, and J.~Lu, ``Edge artificial intelligence for {6G}: Vision, enabling technologies, and applications,'' \emph{IEEE J. Sel. Areas Commun.}, vol.~40, no.~1, pp. 5--36, Jan. 2022.

\bibitem{Brown2020GPT3}
T.~Brown, B.~Mann, N.~Ryder, M.~Subbiah, J.~D. Kaplan \emph{et~al.}, ``Language models are few-shot learners,'' in \emph{Adv. Neural Inf. Process. Syst. (NeurIPS)}, vol.~33, May 2020, pp. 1877--1901.

\bibitem{OpenAI2023GPT4}
{OpenAI}, ``{GPT-4 Technical Report},'' \emph{arXiv preprint arXiv:2303.08774}, Mar. 2023.

\bibitem{Chowdhery2023PaLM}
A.~Chowdhery, S.~Narang, J.~Devlin \emph{et~al.}, ``{PaLM}: Scaling language modeling with pathways,'' \emph{J. Mach. Learn. Res.}, vol.~24, no.~1, Jan. 2023.

\bibitem{Meta2024Llama3}
{Llama Team, AI @ Meta}, ``The {Llama} 3 herd of models,'' \emph{arXiv preprint arXiv:2407.21783}, Jul. 2024.

\bibitem{Chu2025Llama3ISCA}
W.~Chu, X.~Xie \emph{et~al.}, ``Scaling {Llama 3} training with efficient parallelism strategies,'' in \emph{Proc. Annu. Int. Symp. Comput. Archit. (ISCA)}, Jun. 2025, pp. 1703--1716.

\bibitem{McMahan2017FedAvg}
B.~McMahan, E.~Moore, D.~Ramage, S.~Hampson, and B.~A.~y. Arcas, ``Communication-efficient learning of deep networks from decentralized data,'' in \emph{Proc. Int. Conf. Artif. Intell. Stat. (AISTATS)}, vol.~54, Apr. 2017, pp. 1273--1282.

\bibitem{Chen2021JointLearningCommFL}
M.~Chen, Z.~Yang, W.~Saad, C.~Yin, H.~V. Poor, and S.~Cui, ``A joint learning and communications framework for federated learning over wireless networks,'' \emph{IEEE Trans. Wireless Commun.}, vol.~20, no.~1, pp. 269--283, Jan. 2021.

\bibitem{Zhu2021OneBitOTA}
G.~Zhu, Y.~Du, D.~G{\"u}nd{\"u}z, and K.~Huang, ``One-bit over-the-air aggregation for communication-efficient federated edge learning: Design and convergence analysis,'' \emph{IEEE Trans. Wireless Commun.}, vol.~20, no.~3, pp. 2120--2135, Mar. 2021.

\bibitem{Fedus2022Switch}
W.~Fedus, B.~Zoph, and N.~Shazeer, ``Switch transformers: Scaling to trillion parameter models with simple and efficient sparsity,'' \emph{J. Mach. Learn. Res.}, vol.~23, no. 120, pp. 1--39, Jan. 2022.

\bibitem{Du2022GLaM}
N.~Du, Y.~Huang, A.~M. Dai \emph{et~al.}, ``{GLaM}: Efficient scaling of language models with mixture-of-experts,'' in \emph{Proc. Int. Conf. Mach. Learn. (ICML)}, vol. 162, Jul. 2022, pp. 5547--5569.

\bibitem{Shazeer2017MoE}
N.~Shazeer, A.~Mirhoseini, K.~Maziarz, A.~Davis, Q.~V. Le, G.~Hinton, and J.~Dean, ``Outrageously large neural networks: The sparsely-gated mixture-of-experts layer,'' in \emph{Proc. Int. Conf. Learn. Represent. (ICLR)}, Apr. 2017.

\bibitem{Xue2024WDMoE}
N.~Xue, Y.~Sun, Z.~Chen, M.~Tao, X.~Xu, L.~Qian, S.~Cui, W.~Zhang, and P.~Zhang, ``{WDMoE}: Wireless distributed mixture of experts for large language models,'' \emph{IEEE Trans. Wireless Commun.}, vol.~25, pp. 559--572, 2026.

\bibitem{Chen2025SlimCaching}
Q.~Chen, X.~Chen, and K.~Huang, ``{SlimCaching}: Edge caching of mixture-of-experts for distributed inference,'' \emph{arXiv preprint arXiv:2507.06567}, Jul. 2025.

\bibitem{You2023HierPersonalizedTWC}
C.~You, K.~Guo, H.~H. Yang, and T.~Q.~S. Quek, ``Hierarchical personalized federated learning over massive mobile edge computing networks,'' \emph{IEEE Trans. Wireless Commun.}, vol.~22, no.~11, pp. 8141--8157, Nov. 2023.

\bibitem{Aygun2024HierClusteringTWC}
O.~Ayg{\"u}n, M.~Kazemi, D.~G{\"u}nd{\"u}z, and T.~M. Duman, ``Over-the-air federated edge learning with hierarchical clustering,'' \emph{IEEE Trans. Wireless Commun.}, vol.~23, no.~12, pp. 17\,856--17\,871, Dec. 2024.

\bibitem{Ntontin2025Space6GProcIEEE}
K.~Ntontin, E.~Lagunas, J.~Querol, J.~u. Rehman, J.~Grotz, S.~Chatzinotas, and B.~Ottersten, ``A vision, survey, and roadmap toward space communications in the 6g and beyond era,'' \emph{Proc. IEEE}, vol. 113, no.~9, pp. 987--1023, Sept. 2025.

\bibitem{Le2025RandomAccessCOMST}
T.~T.~T. Le, N.~U. Hassan, X.~Chen, M.-S. Alouini, Z.~Han, and C.~Yuen, ``A survey on random access protocols in direct-access {LEO} satellite-based {IoT} communication,'' \emph{IEEE Commun. Surveys Tuts.}, vol.~27, no.~1, pp. 426--462, Feb. 2025.

\bibitem{Matthiesen2024SatFLNetwork}
B.~Matthiesen, N.~Razmi, I.~Leyva-Mayorga, A.~Dekorsy, and P.~Popovski, ``Federated learning in satellite constellations,'' \emph{IEEE Netw.}, vol.~38, no.~2, pp. 232--239, Mar. 2024.

\bibitem{Razmi2024OnboardISLTCOM}
N.~Razmi, B.~Matthiesen, A.~Dekorsy, and P.~Popovski, ``On-board federated learning for satellite clusters with inter-satellite links,'' \emph{IEEE Trans. Commun.}, vol.~72, no.~6, pp. 3408--3424, Jun. 2024.

\bibitem{Fang2023OliveBranchTWC}
Q.~Fang, Z.~Zhai, S.~Yu, Q.~Wu, X.~Gong, and X.~Chen, ``Olive branch learning: A topology-aware federated learning framework for space-air-ground integrated network,'' \emph{IEEE Trans. Wireless Commun.}, vol.~22, no.~7, pp. 4534--4551, Jul. 2023.

\bibitem{Shi2024SatFEEL}
Y.~Shi, L.~Zeng, J.~Zhu, Y.~Zhou, C.~Jiang, and K.~B. Letaief, ``Satellite federated edge learning: Architecture design and convergence analysis,'' \emph{IEEE Trans. Wireless Commun.}, vol.~23, no.~10, pp. 15\,212--15\,229, Oct. 2024.

\bibitem{Zhai2024FedLEO}
Z.~Zhai, Q.~Wu, S.~Yu, R.~Li, F.~Zhang, and X.~Chen, ``{FedLEO}: An offloading-assisted decentralized federated learning framework for low earth orbit satellite networks,'' \emph{IEEE Trans. Mobile Comput.}, vol.~23, no.~5, pp. 5260--5279, May 2024.

\bibitem{Huang2024HFL_SAGIN}
C.~Huang, G.~Chen, P.~Xiao, J.~A. Chambers, and W.~Huang, ``Fair resource allocation for hierarchical federated edge learning in space-air-ground integrated networks via deep reinforcement learning with hybrid control,'' \emph{IEEE J. Sel. Areas Commun.}, vol.~42, no.~12, pp. 3618--3631, Dec. 2024.

\bibitem{Guo2020PFLMoE}
B.~Guo, Y.~Mei, D.~Xiao, and W.~Wu, ``{PFL-MoE}: Personalized federated learning based on mixture of experts,'' in \emph{Web and Big Data (APWeb-WAIM 2021)}, Aug. 2021, pp. 480--486.

\bibitem{Lepikhin2021GShard}
D.~Lepikhin, H.~Lee, Y.~Xu, D.~Chen, O.~Firat, Y.~Huang, M.~Krikun, N.~Shazeer, and Z.~Chen, ``{GShard}: Scaling giant models with conditional computation and automatic sharding,'' in \emph{Proc. Int. Conf. Learn. Represent. (ICLR)}, May 2021.

\bibitem{Cai2025SurveyMoE}
W.~Cai, J.~Jiang, F.~Wang, J.~Tang, S.~Kim, and J.~Huang, ``A survey on mixture of experts in large language models,'' \emph{IEEE Trans. Knowl. Data Eng.}, vol.~37, no.~7, pp. 3896--3915, Jul. 2025.

\bibitem{Qin2025OptimalExpertSelectionDMoE}
S.~Qin, H.~Wu, H.~Du, and K.~Huang, ``Optimal expert selection for distributed mixture-of-experts at the wireless edge,'' \emph{arXiv preprint arXiv:2503.13421}, 2025.

\bibitem{Hu2022LoRA}
E.~J. Hu, Y.~Shen, P.~Wallis, Z.~Allen-Zhu, Y.~Li, S.~Wang, L.~Wang, W.~Chen \emph{et~al.}, ``{LoRA}: Low-rank adaptation of large language models,'' in \emph{Proc. Int. Conf. Learn. Represent. (ICLR)}, vol.~1, no.~2, 2022, p.~3.

\bibitem{Zoph2022STMoE}
B.~Zoph, I.~Bello, S.~Kumar, N.~Du, Y.~Huang, J.~Dean, N.~Shazeer, and W.~Fedus, ``{ST-MoE}: Designing stable and transferable sparse expert models,'' \emph{arXiv preprint arXiv:2202.08906}, Feb. 2022.

\bibitem{meta_llama3_8b_instruct_2024}
{Meta AI}, ``Meta-llama-3-8b-instruct,'' [Online]. Available: \url{https://huggingface.co/meta-llama/Meta-Llama-3-8B-Instruct}, Apr. 2024.

\bibitem{wang_llama3_8b_chinese_chat_2024}
S.~Wang, Y.~Zheng, G.~Wang, S.~Song, and G.~Huang, ``{Llama3-8B-Chinese-Chat},'' [Online]. Available: \url{https://huggingface.co/shenzhi-wang/Llama3-8B-Chinese-Chat}, Mar. 2024.

\bibitem{pal_llama3_openbiollm_8b_2024}
A.~Pal, ``{Llama3-OpenBioLLM-8B},'' [Online]. Available: \url{https://huggingface.co/aaditya/Llama3-OpenBioLLM-8B}, Apr. 2024.

\bibitem{hendrycks2021measuring}
D.~Hendrycks, C.~Burns, S.~Basart, A.~Zou, M.~Mazeika, D.~Song, and J.~Steinhardt, ``Measuring massive multitask language understanding,'' in \emph{Proc. Int. Conf. Learn. Represent. (ICLR)}, May 2021.

\bibitem{liu2023cmexam}
J.~Liu, P.~Zhou, Y.~Hua, D.~Chong, Z.~Tian, A.~Liu, H.~Wang, C.~You, Z.~Guo, L.~Zhu, and M.~L. Li, ``Benchmarking large language models on {CMExam} - a comprehensive chinese medical exam dataset,'' in \emph{Adv. Neural Inf. Process. Syst. (NeurIPS)}, Dec. 2023.

\bibitem{pal2022medmcqa}
A.~Pal, L.~K. Umapathi, and M.~Sankarasubbu, ``{MedMCQA}: A large-scale multi-subject multi-choice dataset for medical domain question answering,'' in \emph{Proc. Conf. Health Inference Learn. (CHIL)}, vol. 174, Apr. 2022, pp. 248--260.

\bibitem{Komatsuzaki2023SparseUpcycling}
A.~Komatsuzaki, J.~Puigcerver, J.~Lee-Thorp, C.~Riquelme~Ruiz, B.~Mustafa, J.~Ainslie, Y.~Tay, M.~Dehghani, and N.~Houlsby, ``Sparse upcycling: Training mixture-of-experts from dense checkpoints,'' in \emph{Proc. Int. Conf. Learn. Represent. (ICLR)}, May 2023.

\bibitem{wang2019glue}
A.~Wang, A.~Singh, J.~Michael, F.~Hill, O.~Levy, and S.~R. Bowman, ``{GLUE}: A multi-task benchmark and analysis platform for natural language understanding,'' in \emph{Proc. Int. Conf. Learn. Represent. (ICLR)}, May 2019.

\end{thebibliography}

\end{document}